\documentclass[a4paper,10pt]{article}
\usepackage[utf8]{inputenc}

\usepackage[a4paper,bindingoffset=0.2in,%
            left=1in,right=1in,top=1in,bottom=1in,%
            footskip=.25in]{geometry}

\usepackage[numbers]{natbib}
\usepackage{todonotes}
\usepackage[hyperindex,breaklinks]{hyperref}
\hypersetup{
  colorlinks   = true,    
  urlcolor     = blue,    
  linkcolor    = blue,    
  citecolor    = red      
}

\usepackage{amsmath,amssymb, amstext}
\usepackage{amsthm}
\usepackage{bm}
\usepackage[font={small,it}]{caption,subcaption}
\usepackage{graphicx}
\usepackage{tikz}
\usetikzlibrary{arrows,arrows.meta,fit,matrix,backgrounds,decorations.pathreplacing,calc}
\usetikzlibrary{positioning}
\usetikzlibrary{decorations.markings}
\usepackage{xspace}
\usepackage{nicefrac}
\usepackage{enumitem}
\usepackage{appendix}
\usepackage[capitalize]{cleveref}

\newcommand\hl[1]{#1}

\usepackage[capitalize]{cleveref}

\theoremstyle{plain}

\newtheorem{obs}{Observation}
\newtheorem{con}{Condition}

\makeatletter
\long\def \protected@iwrite#1#2#3{%
     \begingroup
     \let\thepage\relax
     #2%
     \let\protect\@unexpandable@protect
     \edef\reserved@a{\immediate\write#1{#3}}%
     \reserved@a
     \endgroup
     \if@nobreak\ifvmode\nobreak\fi\fi
    }
\newcounter{numquote}

\makeatother
\newcommand\quoteref[1]{\csname#1\endcsname}

\newlist{steps}{enumerate}{10}
\setlist[steps]{label*=\arabic*)}

\crefname{step}{Step}{Steps}
\crefname{appsec}{Appendix}{Appendices}
\crefname{section}{Section}{Sections}
\crefname{con}{Condition}{Conditions}
\crefname{obs}{Observation}{Observations}
\crefformat{footnote}{#2\footnotemark[#1]#3}

\newif\ifshow
\newcommand*{\diff}{\mathop{\kern0pt\mathrm{d}}\!{}}
\DeclareMathOperator*{\argmax}{arg\,max}

\DeclareMathOperator{\KL}{D_{KL}}

\newcommand{\tx}{{\tilde{x}}}
\newcommand{\tpsi}{{\tilde{\psi}}}
\newcommand{\ts}{{\tilde{s}}}
\newcommand{\ta}{{\tilde{a}}}
\newcommand{\tr}{{\tilde{r}}}

\newcommand{\tPsi}{{\tilde{\Psi}}}

\newcommand{\tomega}{{\tilde{\omega}}}

\newcommand{\bp}{{\bar{p}}}
\newcommand{\bM}{{\bar{M}}}
\newcommand{\bQ}{{\bar{Q}}}
\newcommand{\bGamma}{{\bar{\Gamma}}}
\newcommand{\ddt}{{\frac{d}{d t}}}

\title{A Technical Critique of Some Parts of the Free Energy Principle
}
\author{Martin Biehl$^1$\footnote{These authors contributed equally to this work.}\addtocounter{footnote}{-1}\addtocounter{Hfootnote}{-1} 
 \and Felix A. Pollock$^2$\footnotemark \and Ryota Kanai$^1$
}
\date{%
    $^1$Araya Inc., Tokyo 105-0003, Japan\\%
    $^2$School of Physics and Astronomy, Monash University, Clayton, Victoria 3800, Australia\\[2ex]%
    \today
}

\showtrue  

\crefalias{section}{step}

\begin{document}

\maketitle

\abstract{
We summarize {the original formulation} of the free energy principle, and highlight some technical issues. We discuss how these issues affect related results involving generalised coordinates and, where appropriate, mention consequences for and reveal, up to now unacknowledged, differences to newer formulations of the free energy principle. 
In particular, we reveal that various definitions of the ``Markov blanket'' proposed in different works are not equivalent. We show that crucial steps in the free energy argument which involve rewriting the equations of motion of systems with Markov blankets, are not generally correct without additional (previously unstated) assumptions. We prove by counterexample that the {original} free energy lemma, when taken at face value, is wrong. We show further that this free energy lemma, when it does hold, implies equality of variational density and ergodic conditional density. The interpretation in terms of Bayesian inference hinges on this point, and we hence conclude that it is not sufficiently justified. Additionally, we highlight that the variational densities presented in newer formulations of the free energy principle { and lemma} are parameterised by different variables than in older works, leading to a substantially different interpretation of the theory. Note that we only highlight some specific problems in the discussed publications. These problems do not rule out conclusively that the general ideas behind the free energy principle are worth pursuing.
}

\section*{Overview}

In \citep{friston_life_2013} it is argued that the internal coordinates of an ergodic random dynamical system with a Markov blanket necessarily appear to engage in active Bayesian inference. Here, we reproduce the argument supporting this interpretation in detail and highlight at which points it faces technical issues. In~the course of our critique, we also mention issues of some closely related alternative arguments. In~cases where our results have clear consequences for the more recent related publications \citep{friston_free_2019,parr_markov_2019} we also mention those. In~particular, we point out a conceptual difference in these latter works that has not previously been acknowledged. However, our analysis thereof does not go beyond a few remarks. In~an additional section we discuss the effect of our argument on \citep{friston_cognitive_2014}. The~logical structure of the present paper is depicted in Figure~\ref{fig:1}. 
We note that the technical issues presented here do not affect the validity of approaches where an (expected) free energy minimizing agent is assumed {a priori}, as~presented in, e.g.,~\citep{friston_active_2015}. None of \citep{friston_life_2013,friston_cognitive_2014,friston_free_2019,parr_markov_2019} make this assumption, they instead aim to identify the conditions under which such agents will emerge {within} a given stochastic process. We criticize specific formal issues in the latter publications but leave open whether they can be fixed.
We now briefly introduce the setting of \citep{friston_life_2013} and then sketch the content of this~paper.

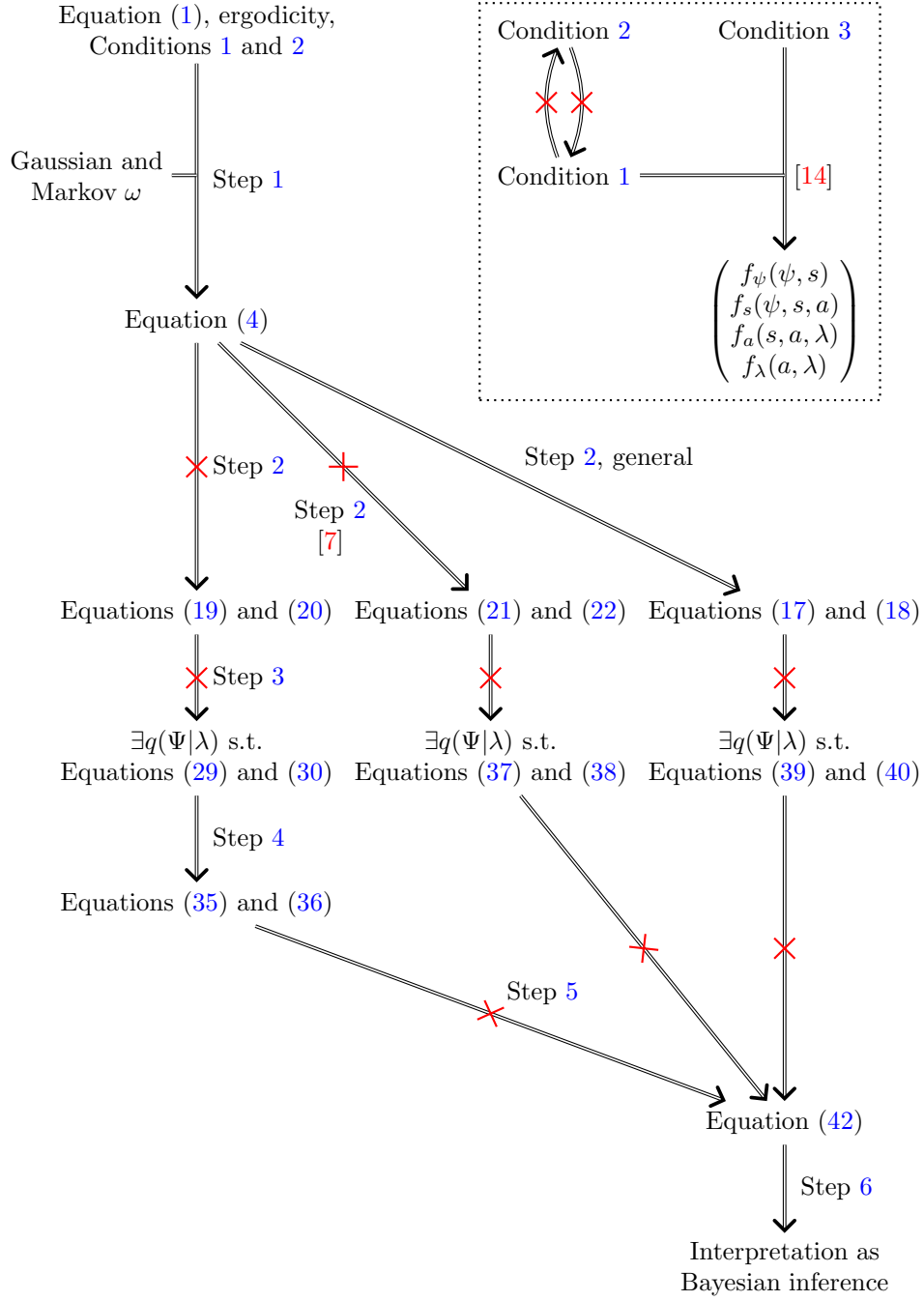
\begin{figure}%
\begin{center}
\begin{tikzpicture}      
				[->,>=angle 90,auto,node distance=2cm]
		\tikzset{
         strike through/.append style={
    decoration={markings, mark=at position 0.5 with {
    \draw[-,thick,red] ++ (4pt,-4pt) -- (-4pt,4pt);
     \draw[-,thick,red] ++ (-4pt,-4pt) -- (4pt,4pt);}
  },postaction={decorate}}
}

		\node (0) [align=center] {Equation (\ref{eq:langevin}), ergodicity,\\ \cref{con:flow,con:fact}};

		\node (c2) [right of=0,node distance=5cm,align=center] {\cref{con:fact}};		
		\node (c1) [below of=c2,align=center] {\cref{con:flow}};
		\node (c3) [right of=c2,node distance=3cm,align=center] {\cref{con:sol}};
		\node (c35) [coordinate,node distance=3cm,right of=c1] {};
		\node (dummy) [coordinate,node distance=3cm,right of=c35] {};
		\node (tdyn) [below of=c35, align=center] {$\begin{pmatrix}
f_\psi(\psi,s) \\
f_s(\psi,s,a) \\
f_a(s,a,\lambda) \\
f_\lambda(a,\lambda)            
          \end{pmatrix}$};
		
		\node[draw,thick,dotted,fit=(c2) (c1) (c3) (tdyn)] {};
		
		\node (05) [coordinate,below of=0, node distance=2cm] {};
		\node (05z) [align=center,left of=05, node distance=1.5cm] {Gaussian and \\Markov $\omega$};
		\node (1) [below of=0,node distance=4cm] {Equation (\ref{eq:25})};
		\node (2) [below of=1, node distance = 4cm] {Equations (\ref{eq:fa}) and (\ref{eq:flambda})};
		\node (2a) [right of=2,node distance =4cm] {Equations (\ref{eq:fa3}) and (\ref{eq:flambda3})};
		\node (2b) [right of=2a,node distance =4cm] {Equations (\ref{eq:fa2}) and (\ref{eq:flambda2})};
		
		\node (3) [below of=2,align=center, node distance = 2cm] {$\exists q(\Psi|\lambda)$ s.t.\ \\ Equations~(\ref{eq:fafe}) and (\ref{eq:flambdafe})};
		\node (3a) [below of=2a,align=center, node distance = 2cm] {$\exists q(\Psi|\lambda)$ s.t.\ \\  Equations~(\ref{eq:fafe4}) and (\ref{eq:flambdafe4})};
		\node (3b) [below of=2b,align=center,node distance =2cm] {$\exists q(\Psi|\lambda)$ s.t.\ \\  Equations~(\ref{eq:fafe3}) and (\ref{eq:flambdafe3})};
		\node (4) [below of=3] {Equations (\ref{eq:gradakl0}) and (\ref{eq:gradlambdakl0})}; 
		\node (45) [coordinate,below of=3b, node distance=4cm] {};
		\node (5) [below of=45,node distance=1cm] {Equation (\ref{eq:pqequal})};
		\node (6) [below of=5,align=center] {Interpretation as \\Bayesian inference};

		\path
        (0) edge [double] node {\;\cref{itm:rwfullgrad}} (1)
        (05z) edge [-,double] (05)
		(1) edge [double,strike through] node {\;\cref{itm:rwpartgrad}}(2)
		(1) edge [double,strike through,style={pos=.55}] node[below,align=center,xshift=-10pt,yshift=-5pt] {\cref{itm:rwpartgrad}\\\citep{friston_free_2019} } (2a)
		(1) edge [double,style={pos=.55}] node {\cref{itm:rwpartgrad}, general} (2b)
		(2) edge [double,strike through] node {\;\cref{itm:exq}} (3)
		(2a) edge [double,strike through] node {} (3a)
		(2b) edge [double,strike through] node {} (3b)
		(3) edge [double] node {\;\cref{itm:gradvanish}} (4)
		(3a) edge [double,strike through] (5)
		(3b) edge [double,strike through] (5)
		(4) edge [double,strike through] node {\;\cref{itm:pqequal}} (5)
		(5) edge [double] node {\;\cref{itm:interpretation}} (6)
		
		(c1) edge [double,bend left=20,strike through]  (c2)
		(c2) edge [double,bend left=20,strike through]  (c1)
		(c3) edge [double,text width=1cm,style={pos=.64}] node {\citep{parr_markov_2019}} (tdyn)
		(c1) edge [-,double] (c35)
		
		;
		\end{tikzpicture}		                 
		\end{center}
		\caption{\hl{Argument visualization}. 
		Numbers labelling edges indicate corresponding steps in this paper. Struck out edges indicate implications that we prove incorrect. The~main argument in \citep{friston_life_2013} takes the left path. The~box in the top right indicates the relations between \cref{con:flow,con:fact,con:sol} and their role in \citep{parr_markov_2019}. Merged edges indicate a logical AND combination of the parent~nodes.}
\label{fig:1}
\end{figure}

The starting point is a random dynamical system whose evolution is governed by the stochastic differential equation
\begin{align}
\label{eq:langevin}
  \dot{x}=f(x) + \omega,
\end{align}
where the system state $x$ and vector field $f(x)$ are multi-dimensional, and~$\omega$ is a Gaussian noise term. There is an additional assumption that the system is ergodic, such that the steady state probability density $p^*(x)$ is well defined. \hl{(In the original paper, the~ergodic density is simply denoted $p(x)$. We here add a star to highlight that it is a time independent probability density.)} 
In this case, $-\ln p^*(x)$ plays the role of a potential function, in~the sense that $f$ can be formulated in terms of its gradients~\citep{ao_potential_2004, kwon_structure_2005}.

It is then assumed that there is a coordinate system $x=(\psi,s,a,\lambda)$ with $\psi=(\psi_1,...,\psi_{n_\psi})$, $s=(s_1,...,s_{n_s})$, $a=(a_1,...,a_{n_a})$, and~$\lambda=(\lambda_1,...,\lambda_{n_\lambda})$, referred to as external, sensory, active, and~internal coordinates \hl{(these are called ``states'' in} \citep{friston_life_2013}) 
 respectively, such that the following condition holds:

\begin{con} \label{con:flow}
The function $f(x)$ can be written as
\begin{align}
\begin{split}
\label{eq:mblanket}
f(x)&= \begin{pmatrix}
f_\psi(\psi,s,a) \\
f_s(\psi,s,a) \\
f_a(s,a,\lambda) \\
f_\lambda(s,a,\lambda)            
          \end{pmatrix}.
  \end{split}
\end{align}
\end{con}

This particular structure is described as ``[formalizing] the dependencies implied by the Markov blanket''~\citep{friston_life_2013}. 
In contrast, more recent works~\citep{friston_free_2019,parr_markov_2019} formulate the Markov blanket in terms of statistical dependencies of the ergodic density $p^*(x) = p^*(\psi,s,a,\lambda)$. 
Specifically, the~following condition is presented:
\begin{con} \label{con:fact}
    The ergodic density factorises as
    \begin{gather}
        p^*(\psi,s,a,\lambda) = p^*(\psi|s,a)p^*(\lambda|s,a)p^*(s,a).
    \end{gather}
\end{con}
In other words, the~internal and external coordinates are independently distributed when conditioned on the sensory and active coordinates. This means we have two different formal expressions of what constitutes a Markov blanket in these publications, and~their relationship has not previously been~established.

Taking \cref{con:flow} to hold, the~argument of \citep{friston_life_2013} then proceeds along the following~steps:
\begin{description}
  \item[\cref{itm:rwfullgrad}] \hl{Rewrite} the vector field $f(\psi,s,a,\lambda)$ describing the dynamics of the system in terms of the gradient of negative logarithm of the ergodic density $p^*(\psi,s,a,\lambda)$ of that system.
  \item[\cref{itm:rwpartgrad}] Rewrite the components $f_\lambda(s,a,\lambda)$ and $f_a(s,a,\lambda)$ of the vector field $f(\psi,s,a,\lambda)$ in terms of only \textit{{partial gradients}} of the negative logarithm of $p^*(\psi,s,a,\lambda)$.
  \item[\cref{itm:exq}] Assert (in the \textit{{Free Energy Lemma}}) the existence of a density $q(\psi|\lambda)$ over the external coordinates $\psi$ parameterized by the internal coordinates $\lambda$, and~that $f(\psi,s,a,\lambda)$ can again be rewritten, this time in terms of a free energy depending on $q(\Psi|\lambda)$. \hl{(Here, and~whenever it would otherwise be ambiguous, we use a capitalized $\Psi$ to indicate full distributions, rather than the probability density for specific value of $\psi$.)} 
  \item[\cref{itm:gradvanish}] Claim 
  that equivalence of the equations of motion in \cref{itm:rwpartgrad} and \cref{itm:exq} implies that certain partial gradients  of the KL divergence between $q(\Psi|\lambda)$ and the conditional ergodic density $p^*(\Psi|s,a,\lambda)$ must vanish.
  \item[\cref{itm:pqequal}] Claim that it follows from \cref{itm:gradvanish} that $q(\Psi|\lambda)$ and $p^*(\Psi|s,a,\lambda)$ are ``rendered'' equal.
  \item[\cref{itm:interpretation}] Interpret 
  \begin{itemize}
\item 
$p^*(\Psi|s,a,\lambda)$ as a posterior over external coordinates given particular values of sensor, active, and~internal coordinates, 
\item $q(\Psi|\lambda)$ as encoding Bayesian beliefs about the external coordinates by the internal coordinates, and~
\item their equality as the internal coordinates appearing to ``solve the problem of Bayesian inference''.                   \end{itemize}
\end{description}

In the present paper, 
we make the following main~observations 
\begin{itemize}
\item The re-expression of Equation~(\ref{eq:langevin}) in the form chosen in \cref{itm:rwfullgrad} is derived under restrictive assumptions, including that the system is subject to Gaussian and Markov~noise.
  \item \cref{con:flow} and \cref{con:fact} are independent from each other.
  \item \cref{con:flow} and \cref{con:sol} together lead to a system where the interpretation of $s$ and $a$ as sensory and active coordinates is questionable. 
  \item Under both \cref{con:flow,con:fact}, the~expressions of $f_\lambda(s,a,\lambda)$ and $f_a(s,a,\lambda)$ resulting from \cref{itm:rwpartgrad} are not as general as those contained in the result of \cref{itm:rwfullgrad}. The~more general alternative expression derived in \citep{friston_free_2019} remains insufficiently general.    
  \item Under both \cref{con:flow,con:fact}, the~Free Energy Lemma, when taken at face value, is wrong and cannot be salvaged by using alternatives in \cref{itm:rwpartgrad}.
  \item Under both \cref{con:flow,con:fact}, contrary to \cref{itm:pqequal} the vanishing of the gradient of the KL divergence does not imply equality of $q(\Psi|\lambda)$ and $p^*(\Psi|s,a,\lambda)$.
  \item As a consequence, the~basic preconditions for the interpretations in \cref{itm:interpretation} are not implied by either of the two proposed Markov blanket \cref{con:flow,con:fact}.
\end{itemize}

The later \citep{friston_cognitive_2014} presents an argument almost identical to the one in the original \citep{friston_life_2013}. In~\cref{sec:ieeepaper} we discuss how our observations apply to this~publication.

\section{Expression via the Gradient of the Ergodic~Density}
\label{itm:rwfullgrad}
Here we introduce the expression of the system's dynamics Equation~(\ref{eq:langevin}) in the form used for the Free Energy Lemma (Lemma 2.1 in \citep{friston_life_2013}).
This form expresses the dynamics of internal and active coordinates of the given ergodic random dynamical system in terms of the gradient of the ergodic density $p^*(x)$.
In accordance with the results of \citep{kwon_structure_2005}, $f(x)$ is rewritten as (see Equation~(2.5) in  \citep{friston_life_2013}):
\begin{align}
\label{eq:25}
    f(x)&=(\Gamma+R) \cdot \nabla \ln p^*(x),
\end{align}
where $\Gamma$ is the diffusion matrix, which we will take to be block diagonal, {(in \citep{friston_life_2013}, and~later work such as \citep{friston_free_2019}, $\Gamma$ \hl{ is taken to be proportional to the identity matrix})} 
and $R$ is an antisymmetric matrix, defined through the relation
\begin{gather} \label{eq:Rdef}
    MR + RM^T = M\Gamma-\Gamma M^T,
\end{gather} 
with 
\begin{gather}
    M_{ij} = \nabla_j f_i(x).
\end{gather} 
Here, and~in all of \citep{friston_life_2013,friston_cognitive_2014,friston_free_2019,parr_markov_2019}, both $\Gamma$ and $R$ are assumed constant. We emphasise here that, for~general nonlinear models, these matrices can vary with the coordinates and \linebreak  Equation~(\ref{eq:Rdef}) holds only approximately \citep{kwon_nonequilibrium_2011,ma_complete_2015}. {(\hl{The exact conditions under which these} matrices can be chosen to be constant can be found in~\citep{ma_complete_2015,yuan_sde_2017} and, \hl{for~the discrete state case,}~\citep{ao_dynamical_2013}}). 
Moreover, Equation~(\ref{eq:25}) is derived in the literature under the explicit assumption that the fluctuations $\omega$ be Gaussian and Markov~\citep{ao_potential_2004, kwon_structure_2005}.  For~the counterexamples we present here, we restrict ourselves to the class of Ornstein-Uhlenbeck processes, for~which $R$ and $\Gamma$ are always constant, and~the ergodic density $p^*(x)=p^*(\psi,s,a,\lambda)$ is necessarily a multivariate Gaussian with zero mean. Specifically, following~\citep{kwon_structure_2005},
\begin{align} \label{eq:ergdef}
  p^*(\psi,s,a,\lambda):=\frac{1}{Z} \exp\left[{-\frac{1}{2} (\psi,s,a,\lambda)U(\psi,s,a,\lambda)^\top}\right],
\end{align}
where $(\psi,s,a,\lambda)$ is a row vector and $Z$ is a suitable normalisation constant. From~ \linebreak Equation~(\ref{eq:25}) it can be seen that,
\begin{align} \label{eq:Udef}
  U=-(\Gamma + R)^{-1} M;
\end{align}
though we emphasise here that strict relations between $M$ and $U$ can only be made because of the assumption that $\Gamma$ and $R$ are coordinate independent~\citep{yuan_lyapunov_2014}. This concludes \cref{itm:rwfullgrad}.

Before moving on to \cref{itm:rwpartgrad}, we note that, under~the assumptions implicit in \cref{itm:rwfullgrad}, we can express \cref{con:flow} and \cref{con:fact} in terms of the matrices $M$ and $U$. {\hl{(In} the nonlinear case, these matrices can still be defined in terms of derivatives of the force vector field and potential, respectively; however, they will be generally coordinate-dependent, even when $\Gamma$ and $R$ are not~\citep{kwon_nonequilibrium_2011}.\hl{)}} 
Firstly, since it effectively states that $\nabla_\psi f_a(x)=\nabla_\psi f_\lambda(x)=\nabla_\lambda f_s(x)=\nabla_\lambda f_\psi(x)=0$,
\begin{gather} \label{eq:con1mat}
    {\rm \cref{con:flow}} \;\;\Leftrightarrow\;\; M_{a\psi} = M_{\lambda \psi} = M_{s\lambda} = M_{\psi \lambda} = 0,
\end{gather}
with $M_{\alpha \beta}$ a block sub-matrix of $M$ in general. Secondly, because~of the multivariate Gaussian nature of $p^*(\psi,s,a,\lambda)$, the~dependencies of conditional distributions are encoded in the inverse $U$ of the covariance matrix; we therefore have that
\begin{gather} \label{eq:con2mat}
    {\rm \cref{con:fact}} \;\;\Leftrightarrow\;\; U_{\psi \lambda}=U_{\lambda \psi}=0,
\end{gather}
where $U_{\alpha\beta}$ is a block sub-matrix of $U$. These implications bring us to our first observation:
\begin{obs} \label{obs:condcounter}
 Neither one of \cref{con:flow} (the vector field dependency structure) or \cref{con:fact} (conditional independence in the ergodic distribution) implies the other:
\begin{align}
     &{\rm \cref{con:flow}}\;\; \nRightarrow \;\;{\rm \cref{con:fact}} \label{eq:flownotimplyfact}\\
     &{\rm \cref{con:flow}}\;\; \nLeftarrow \;\;{\rm \cref{con:fact}}\label{eq:factnotimplyflow}.      
 \end{align}
\end{obs}
\begin{proof}
In \cref{app:conditioncounter}, we provide direct counterexamples, using the equivalent constraints on the matrices $M$ and $U$ in Equations~(\ref{eq:con1mat}) and (\ref{eq:con2mat}), to~implication in either direction. That is, there exists a system obeying \cref{con:flow} that does not obey \cref{con:fact} (proving Equation~(\ref{eq:flownotimplyfact})), and~there exists one obeying \cref{con:fact} that does not obey \cref{con:flow} (proving Equation~(\ref{eq:factnotimplyflow})).
\end{proof}
Henceforth, unless~otherwise stated, we will assume both \cref{con:flow}~and \linebreak \cref{con:fact}. Any implications that fail to hold in this special case cannot hold~generally.

\section{Re-Expression Using only Partial~Gradients}
\label{itm:rwpartgrad}
For \cref{itm:rwpartgrad} we focus on the components $f_\lambda=(f_{\lambda_1},...,f_{n_\lambda})$ and $f_a=(f_{a_1},...,f_{n_a})$ of $f$. Without~loss of generality we can rewrite them from  Equation~(\ref{eq:25}) as:
\begin{align}
f_a(s,a,\lambda) =& \left( R_{a \psi} \cdot \nabla_\psi + R_{a s} \cdot \nabla_s+ (\Gamma_{a a}+ R_{a a}) \cdot \nabla_a \right. \nonumber\\ &\;\left.+ R_{a \lambda} \cdot \nabla_\lambda \right) \ln p^*(\psi,s,a,\lambda) , \label{eq:fa0}\\
f_\lambda(s,a,\lambda) =&\left( R_{\lambda \psi} \cdot \nabla_\psi + R_{\lambda s} \cdot \nabla_s  + (\Gamma_{\lambda \lambda}+R_{\lambda \lambda}) \cdot \nabla_\lambda \right. \nonumber\\ &\;\left. + R_{\lambda a} \cdot \nabla_a\right) \ln p^*(\psi,s,a,\lambda),  \label{eq:flambda0}
\end{align}
where $\Gamma_{nm}$ ($R_{nm}$) is the block of $\Gamma$ ($R$) connecting derivatives with respect to the $m$ coordinates to the time derivatives of the $n$ coordinates.
The expectation value with respect to $p^*(\psi|s,a,\lambda)$ leaves the left hand side of these equations unchanged.  A~few manipulations (\citep{friston_free_2019} cf. Equation (12.14), p. 129) reveal that, on~the right hand side, this leads to  the ergodic density $p^*(\psi,s,a,\lambda)$ being replaced by the marginalised ergodic density $p^*(s,a,\lambda)$ so that we get
\begin{align}
f_a(s,a,\lambda) =& \left( R_{a \psi} \cdot \nabla_\psi + R_{a s} \cdot \nabla_s + (\Gamma_{a a}+ R_{a a}) \cdot \nabla_a \right. \nonumber\\ &\;\left.+ R_{a \lambda} \cdot \nabla_\lambda \right) \ln p^*(s,a,\lambda)  \label{eq:fa1}\\
f_\lambda(s,a,\lambda) =&\left( R_{\lambda \psi} \cdot \nabla_\psi + R_{\lambda s} \cdot \nabla_s + (\Gamma_{\lambda \lambda}+R_{\lambda \lambda}) \cdot \nabla_\lambda \right. \nonumber\\ &\;\left.+ R_{\lambda a} \cdot \nabla_a  \right) \ln p^*(s,a,\lambda).  \label{eq:flambda1}
\end{align}
Since $\nabla_\psi \ln p^*(s,a,\lambda)=0$, the~terms involving $\nabla_\psi$ drop out:
\begin{align}
f_a(s,a,\lambda)&= \left( R_{a s} \cdot \nabla_s + (\Gamma_{a a}+ R_{a a}) \cdot \nabla_a + R_{a \lambda} \cdot \nabla_\lambda \right) \ln p^*(s,a,\lambda),  \label{eq:fa2}\\
f_\lambda(s,a,\lambda) &=\left( R_{\lambda s} \cdot \nabla_s + R_{\lambda a} \cdot \nabla_a + (\Gamma_{\lambda \lambda}+R_{\lambda \lambda}) \cdot \nabla_\lambda \right) \ln p^*(s,a,\lambda).  \label{eq:flambda2}
\end{align}
We are not aware of how to further simplify this equation without additional assumptions. However, in~(Equations (2.5) and (2.6) of \citep{friston_life_2013} ) all of the off-diagonal terms are implicitly assumed to vanish, i.e.,~ Equation~(\ref{eq:25}) is equated with:
\begin{align}
f_a(s,a,\lambda)&= (\Gamma_{a a}+ R_{a a}) \cdot \nabla_a \ln p^*(s,a,\lambda),  \label{eq:fa}\\
f_\lambda(s,a,\lambda) &= (\Gamma_{\lambda \lambda}+R_{\lambda \lambda}) \cdot \nabla_\lambda \ln p^*(s,a,\lambda).  \label{eq:flambda}
\end{align}
This equation is the result of \cref{itm:rwpartgrad}. 

In the more recent (Appendix B of \citep{friston_free_2019}) a more detailed discussion of  Equation~(\ref{eq:25}) is presented, where it is claimed that \cref{con:flow} implies \cref{con:fact} (cf. our \cref{obs:condcounter}) along with the following simplification of Equations~(\ref{eq:fa2}) and  (\ref{eq:flambda2})~(\citep{friston_free_2019}, Equations~(12.8)--(12.11), (12.15), pp. 126--129):
\begin{align}
f_a(s,a,\lambda)&= \left((\Gamma_{a a}+ R_{a a}) \cdot \nabla_a +R_{a \lambda} \cdot \nabla_\lambda \right) \ln p^*(s,a,\lambda),   \label{eq:fa3}\\
f_\lambda(s,a,\lambda) &= \left(R_{\lambda a} \cdot \nabla_a +(\Gamma_{\lambda \lambda}+R_{\lambda \lambda}) \cdot \nabla_\lambda \right) \ln p^*(s,a,\lambda).  \label{eq:flambda3}
\end{align}
 However, Equations (\ref{eq:fa3}) and (\ref{eq:flambda3}) are still provably less general than Equations~(\ref{eq:fa0}) and \linebreak   (\ref{eq:flambda0}), even when both \cref{con:flow} and \cref{con:fact} are~satisfied. 

\begin{obs} \label{obs:partgradcounter}
Given a random dynamical system obeying Equation~(\ref{eq:langevin}), ergodicity, and~both \cref{con:flow} and \cref{con:fact}, none of Equations~(\ref{eq:fa})--(\ref{eq:flambda3}) generally hold.
\end{obs}
\begin{proof}
By counterexample, see \cref{app:step2counter}. There, we show explicitly that a model satisfying the above assumptions does not satisfy the equations in question.
\end{proof}

In order to arrive at Equations~(\ref{eq:fa3}) and  (\ref{eq:flambda3}) from Equations~(\ref{eq:fa2}) and  (\ref{eq:flambda2}) in general, one must remove the offending ``solenoidal flow'' terms by fiat. That is, one assumes $R_{as} = R_{\lambda s} = 0$. In~(\citep{friston_free_2019}, Equation (12.4)), the~following, even stronger, condition is assumed as an alternative starting point (along with \cref{con:fact}):
\begin{con} \label{con:sol}
    The blocks of the $R$ matrix appearing in~Equation (\ref{eq:25}) coupling $(s,a)$ coordinates to $\lambda$ and $\psi$ coordinates, and~$\psi$ coordinates to $\lambda$ coordinates vanish, i.e.,
    \begin{gather} \label{eq:nosolenoidal}
        R_{\psi s} = R_{\psi a} = R_{\psi \lambda} = R_{s \lambda} = R_{a \lambda} =0.
    \end{gather}
\end{con}
This is claimed to imply $M_{\psi \lambda} = M_{\lambda \psi} = 0$, but~not the full \cref{con:flow}. However, in~\citep{parr_markov_2019}, both \cref{con:flow} and \cref{con:sol} are assumed (along with $R_{as} = 0$). This prompts our next~observation.

\begin{obs}
In a system satisfying both \cref{con:flow} and \cref{con:sol}, the~internal coordinates cannot be directly influenced by the sensory coordinates: $f_\lambda(s,a,\lambda) = f_\lambda(a,\lambda)$, and~the external coordinates cannot be directly influenced by the active coordinates: $f_\psi(\psi,s,a) = f_\psi(\psi,s)$.
\end{obs}
\begin{proof}
From Equation~(\ref{eq:Rdef}), it follows that
\begin{gather}
     M = (\Gamma + R) M^T (\Gamma - R)^{-1},
\end{gather}
with the inverse replaced by a pseudoinverse if $\Gamma-R$ is not invertible.
Therefore, if \linebreak $\Gamma_{\alpha \beta} = \delta_{\alpha \beta} \Gamma_{\alpha \alpha}$ and $R_{\alpha \beta} = \delta_{\alpha \beta} R_{\alpha \alpha}$ for blocks of coordinates labelled by $\alpha$ and $\beta$, then
\begin{gather} \label{eq:Mblocks}
    M_{\alpha\beta} = (\Gamma_{\alpha \alpha} + R_{\alpha \alpha}) M^T_{\beta\alpha} (\Gamma_{\beta \beta} - R_{\beta \beta})^{-1},
\end{gather}
and $M_{\beta\alpha}=0 \Rightarrow M_{\alpha\beta}=0$.

\cref{con:sol} implies that the only nonzero blocks of $R$ are $R_{\psi\psi}$, $R_{ss}$, $R_{sa}$, $R_{as}$, $R_{aa}$, and~$R_{\lambda \lambda}$, and~$\Gamma$ is assumed to be block diagonal. As~noted in Equation~(\ref{eq:con1mat}), \cref{con:flow} requires that $M_{a\psi} = M_{\lambda \psi} = M_{s\lambda} = M_{\psi \lambda} = 0$. Through Equation~(\ref{eq:Mblocks}), these together imply that $M_{\lambda s}=M_{\psi a}=0$, and~hence that
\begin{align}
\label{eq:fxparr}
\begin{split}
f(x)&= \begin{pmatrix}
f_\psi(\psi,s) \\
f_s(\psi,s,a) \\
f_a(s,a,\lambda) \\
f_\lambda(a,\lambda)            
          \end{pmatrix},
  \end{split}
\end{align}
as was to be shown.
\end{proof}

In this case, the~four sets of coordinates interact in a chain, and~it is questionable whether the $s$ and $a$ coordinates can be meaningfully interpreted, respectively, as~sensory inputs to the internal coordinates or their boundary-mediated influence on the external coordinates. 

\section{Free Energy~Lemma}
\label{itm:exq}
The relation of the dynamics of the internal coordinates to Bayesian beliefs is made by introducing a density (called the variational density) $q(\Psi|\lambda)$ that is then interpreted as encoding a Bayesian belief. It is parameterized by the internal coordinates $\lambda$ and claimed to be ``arbitrary''. We take this ``at face value'' and consider $q(\Psi|\lambda)$ to be parameterized \emph{{only}} 
   by $\lambda$ and, therefore, to~be independent of $(s,a)$. \hl{(We}   
   note that there is a convention in the literature on variational Bayesian inference 
 e.g.,~in \citep{bishop_pattern_2006} to drop the observed variables/data in the variational density. It is possible that in \citep{friston_life_2013}, $(s,a)$ is seen as observed variables and dropped from the variational density $q(\Psi|\lambda)$ as in this convention. However, the~reason that dropping the observed variables is justified in the established convention is that those observed variables are fixed throughout the minimization of the variational free energy and the parameters of the variational density do not influence the observed data in any way. In~other words the variational density is optimized for a single datapoint. In~\citep{friston_life_2013} the datapoint is continuously changing and partially doing so in dependence on the parameter $\lambda$ as $\dot{a}=f_a(s,a,\lambda)$. These differences and their consequences are non-trivial and beyond the scope of this paper so we assume that the variational density does not depend on $(s,a)$.\hl{)} If $q(\Psi|\lambda)$ is allowed to depend on $(s,a)$, \cref{obs:felcounter} does not apply and the Free Energy Lemma is made trivially true by setting $q(\psi|s,a,\lambda):=p^*(\psi|s,a,\lambda)$. The~existence of the variational density $q(\Psi|\lambda)$ is asserted by the Free Energy Lemma (see Lemma 2.1 in \citep{friston_life_2013}). \hl{(Explicitly, the~Free Energy Lemma asserts the existence of a free energy $F(s,a,\lambda)$ in terms of which $f(\psi,s,a,\lambda)$ can be expressed and not the existence of $q(\Psi|\lambda)$. However, since the free energy is defined as a functional of $q(\Psi|\lambda)$, it exists if and only if a suitable $q(\Psi|\lambda)$ exists.)} 
 
More precisely, the~Free Energy Lemma (and \cref{itm:exq}) asserts that for every ergodic density {\hl{(}equivalently as expressed in \citep{friston_life_2013}, for~every Gibbs energy $G(x):=-\ln p^*(\psi,s,a,\lambda)$\hl{)}}    
$p^*(\psi,s,a,\lambda)$ of a system obeying Equations~(\ref{eq:fa}) and  (\ref{eq:flambda}) there is a free energy $F(s,a,\lambda)$, defined as
\begin{align}
  F(s,a,\lambda):&= -\ln p^*(s,a,\lambda) + \int q(\psi|\lambda) \ln \frac{q(\psi|\lambda)}{p^*(\psi|s,a,\lambda)} \diff \psi \\
  &= -\ln p^*(s,a,\lambda) + \KL[q(\Psi|\lambda)||p^*(\Psi|s,a,\lambda)],\label{eq:fekl}
\end{align}
in terms of the ``posterior density'' $p^*(\Psi|s,a,\lambda)$, {\hl{(}here, we keep the conditioning argument $\lambda$, as~in \citep{friston_life_2013}, and~do not explicitly assume \cref{con:fact}, though~our conclusions are unaffected by it\hl{)}} such that Equations~(\ref{eq:fa}) and  (\ref{eq:flambda}) can be rewritten as:
\begin{align}
f_a(s,a,\lambda) &= -(\Gamma+R)_{a a} \cdot \nabla_a F(s,a,\lambda), \label{eq:fafe}\\
f_\lambda(s,a,\lambda) &= -(\Gamma+R)_{\lambda \lambda} \cdot \nabla_\lambda F(s,a,\lambda)\label{eq:flambdafe}.         
\end{align}

It is worth considering what a proof of the Free Energy Lemma could look like.
A proof of existence of a free energy (and therefore of the Free Energy Lemma) would need to show that, for~every system satisfying the given assumptions, there always exists a $q(\Psi|\lambda)$ such that the right hand sides of Equations~(\ref{eq:fafe}) and  (\ref{eq:flambdafe}) are equal to the right hand sides of Equations~(\ref{eq:fa}) and  (\ref{eq:flambda}). Expanding Equations~(\ref{eq:fafe}) and  (\ref{eq:flambdafe}) using (\ref{eq:fekl}) leads to:
\begin{align}
\begin{split}
f_a(s,a,\lambda) &= (\Gamma+R)_{a a} \cdot \nabla_a \ln p^*(s,a,\lambda) \\
&\phantom{(\Gamma+R)_{a a}}-  (\Gamma+R)_{a a} \cdot \nabla_a \KL[q(\Psi|\lambda)||p^*(\Psi|s,a,\lambda)],\label{eq:fafe2}
\end{split}\\
\begin{split}
f_\lambda(s,a,\lambda) &= (\Gamma+R)_{\lambda \lambda} \cdot \nabla_\lambda \ln p^*(s,a,\lambda)\\ 
&\phantom{(\Gamma+R)_{\lambda \lambda}}-(\Gamma+R)_{\lambda \lambda} \cdot \nabla_\lambda \KL[q(\Psi|\lambda)||p^*(\Psi|s,a,\lambda)].\label{eq:flambdafe2}          
\end{split}
\end{align}
For equality of the right hand sides to those of Equations~(\ref{eq:fa}) and  (\ref{eq:flambda}) we need:
\begin{align}
 (\Gamma+R)_{a a} \cdot \nabla_a \KL[q(\Psi|\lambda)||p^*(\Psi|s,a,\lambda)]&=0\label{eq:nullgrada}
\\
(\Gamma+R)_{\lambda \lambda} \cdot \nabla_\lambda \KL[q(\Psi|\lambda)||p^*(\Psi|s,a,\lambda)]&=0.\label{eq:nullgradlambda}          
\end{align}
In words, these equations say that the Free Energy Lemma holds if any of the following three conditions (of strictly increasing strengths) are~given: 
\begin{enumerate}
    \item There is a $q(\Psi|\lambda)$ such that the partial gradients $\nabla_a$ and $\nabla_\lambda$ of the KL divergence between the variational density and the conditional ergodic density are elements of the nullspaces of $(\Gamma+R)_{a a}$ and $(\Gamma +R)_{\lambda \lambda}$ respectively.
    \item There is a $q(\Psi|\lambda)$ such that the gradients of the KL divergence to $p^*(\Psi|s,a,\lambda)$ are equal to the nullvector:
\begin{align}
  \nabla_a \KL[q(\Psi|\lambda)||p^*(\Psi|s,a,\lambda)] =0,\label{eq:gradakl0}\\
  \nabla_\lambda \KL[q(\Psi|\lambda)||p^*(\Psi|s,a,\lambda)] =0,\label{eq:gradlambdakl0}
\end{align}
Then they are always elements of the nullspaces of $(\Gamma+R)_{a a}$ and $(\Gamma +R)_{\lambda \lambda}$ respectively.
    \item There is a $q(\Psi|\lambda)$ such that $q(\Psi|\lambda)=p^*(\Psi|s,a,\lambda)$ (and hence $p^*(\Psi|s,a,\lambda)=p^*(\Psi|\lambda)$) which implies that the KL divergence to $p^*(\Psi|s,a,\lambda)$ vanishes for all $a,\lambda$ and the two partial gradients are always nullvectors and therefore elements of the according nullspaces. 
\end{enumerate}

The Free Energy Lemma can then be proven by showing that one of these three cases follows from the conditions of the lemma.
However, no attempt is made in \citep{friston_life_2013} to establish this.
Instead the given proof discusses purported \emph{{consequences}} of the existence of a suitable $q(\Psi|\lambda)$. These will be discussed in \cref{itm:gradvanish,itm:pqequal}.
 
Even if the Free Energy Lemma does not hold for systems obeying \linebreak Equations~(\ref{eq:fa}) and  (\ref{eq:flambda}), one might expect that systems that instead only satisfy the more general Equations~(\ref{eq:fa3}) and  (\ref{eq:flambda3}) or the most general   Equations~(\ref{eq:fa2}) and  (\ref{eq:flambda2}). For~these systems the Free Energy Lemma would require that there is a $q(\Psi|\lambda)$ such that
\begin{align}
    f_a(s,a,\lambda)&= \left( (\Gamma_{a a}+ R_{a a}) \cdot \nabla_a + R_{a \lambda} \cdot \nabla_\lambda \right) F(s,a,\lambda), \label{eq:fafe4} \\
f_\lambda(s,a,\lambda) &=\left( R_{\lambda a} \cdot \nabla_a + (\Gamma_{\lambda \lambda}+R_{\lambda \lambda}) \cdot \nabla_\lambda \right) F(s,a,\lambda). \label{eq:flambdafe4}
\end{align}
or
\begin{align}
    f_a(s,a,\lambda)&= \left( R_{a s} \cdot \nabla_s + (\Gamma_{a a}+ R_{a a}) \cdot \nabla_a + R_{a \lambda} \cdot \nabla_\lambda \right) F(s,a,\lambda), \label{eq:fafe3} \\
f_\lambda(s,a,\lambda) &=\left( R_{\lambda s} \cdot \nabla_s + R_{\lambda a} \cdot \nabla_a + (\Gamma_{\lambda \lambda}+R_{\lambda \lambda}) \cdot \nabla_\lambda \right) F(s,a,\lambda), \label{eq:flambdafe3}
\end{align}
hold respectively. However, we find this not to be the case in general.
\begin{obs}
\label{obs:felcounter}
Given a random dynamical system obeying Equation (\ref{eq:langevin}), ergodicity, \linebreak  \cref{con:flow} and \cref{con:fact}, there need not exist a free energy expressed in terms of a variational density  $q(\Psi|\lambda)$ such that: 
\begin{description}
\item[(i)] Equations (\ref{eq:fafe}) and (\ref{eq:flambdafe}) hold if Equations (\ref{eq:fa}) and (\ref{eq:flambda})  do;
 \item[(ii)] Equations (\ref{eq:fafe4}) and (\ref{eq:flambdafe4})   hold if Equations (\ref{eq:fa}) and (\ref{eq:flambda})  don't hold but  Equations (\ref{eq:fa3}) and (\ref{eq:flambda3})  do; 
\item[(iii)]  Equations (\ref{eq:fafe3}) and (\ref{eq:flambdafe3})   hold if neither  Equations (\ref{eq:fa}) and (\ref{eq:flambda}) norEquations (\ref{eq:fa3}) and (\ref{eq:flambda3}) hold but Equations (\ref{eq:fa2}) and (\ref{eq:flambda2}) do.
\end{description}
\end{obs}
\begin{proof}
In \cref{app:felcounter}, we derive a set of conditions on the $R$ and $U$ matrices, and~on the putative variational density $q(\Psi|\lambda)$, that follow from each of the pairs of equations in cases (i--iii). We show that, in~general, each pair leads to a contradiction and, in~each case, provide a counterexample that falls in the according system class. 
\end{proof}

Before proceeding, we note that later works present an alternative version of the Free Energy Lemma, where the conditioning argument of $q(\Psi|\lambda)$ is replaced by the most likely value of $\lambda$ conditional on the $(s,a)$ coordinates~\citep{friston_free_2019,parr_markov_2019}. We here concern ourselves with the version apparent in \citep{friston_life_2013}, where $q(\Psi|\lambda)$ is parameterised by the internal states themselves, but~will briefly comment on the interpretation of the alternative approach in \cref{itm:interpretation}.

\section{Vanishing~Gradients}
\label{itm:gradvanish}
As mentioned in \cref{itm:exq}, the~proof of the Free Energy Lemma in \citep{friston_life_2013} only discusses its consequences. The~first proposed consequence is that expressing the vector field in terms of a free energy as in Equations~(\ref{eq:fafe}) and  (\ref{eq:flambdafe}) ``requires'' that the gradients with respect to $a$ and $\lambda$ of the KL divergence vanish, i.e., that Equations~(\ref{eq:gradakl0}) and  (\ref{eq:gradlambdakl0}) hold. 

We mentioned in \cref{itm:exq} that the implication in the opposite direction holds. This can be seen from Equations~(\ref{eq:nullgrada}) and  (\ref{eq:nullgradlambda}).
However, if~the nullspace of $(\Gamma+R)_{aa}$ or $(\Gamma+R)_{\lambda \lambda}$ is non-trivial
, then the gradient may be a non-zero element of this subspace and \linebreak Equations~(\ref{eq:fafe}) and  (\ref{eq:flambdafe}) will still hold. In~that case the vanishing gradients would not be necessary for the Free Energy~Lemma.

The conditions under which a non-trivial nullspace exists are discussed in \citep{kwon_structure_2005}. In~short, the~nullspace is guaranteed to be trivial in the special case where $\Gamma$ is positive definite. Whether or not ergodic systems with a Markov blanket can ever admit a non-trivial nullspace, and~hence divergences in Equations~(\ref{eq:fafe2}) and  (\ref{eq:flambdafe2}) with non-vanishing gradients, is not immediately clear. However, in~order to establish the necessity of Equations~(\ref{eq:gradakl0}) and  (\ref{eq:gradlambdakl0}) this remains to be~proven.

\section{Equality of $Q( \Psi| \lambda)$ and $P^*( \Psi|s,a, \lambda)$}
\label{itm:pqequal}
The proof of the Free Energy Lemma in \citep{friston_life_2013} also proposes that the vanishing of gradients of the KL divergence, of~the variational density $q(\Psi|\lambda)$ from the conditional ergodic density $p^*(\Psi|s,a,\lambda)$, implies the equality of these densities. 
We mentioned in Equations~(\ref{itm:gradvanish}) that the implication in the opposite direction holds. This can also be seen from \linebreak Equations~(\ref{eq:nullgrada}) and  (\ref{eq:nullgradlambda}). 
Concerning the implication in the direction proposed by \citep{friston_life_2013}, let us now assume that for a given system Equations~(\ref{eq:fa}) and  (\ref{eq:flambda}) hold, a~variational density $q(\Psi|\lambda)$ does exist, and~the gradients of the KL divergence of the variational and ergodic densities vanish i.e., Equations~(\ref{eq:gradakl0}) and  (\ref{eq:gradlambdakl0}) hold. 
Then consider the argument by \citep{friston_life_2013} in this direct quote (comments in square brackets by us): 
\begin{quote}
\hl{``}However, Equation~({2.6}) [Equations (\ref{eq:fa}) and  (\ref{eq:flambda}) above] requires the gradients of the
divergence to be zero [Equations (\ref{eq:gradakl0}) and  (\ref{eq:gradlambdakl0})], which means the divergence must be
minimized with respect to internal states. This means that
the variational and posterior densities must be equal:
\begin{align*}
q(\psi|\lambda)=p^{[*]}(\psi|s,a,\lambda)\Rightarrow \KL=0\Rightarrow\begin{cases}
                                                             &(\Gamma+R) \cdot \nabla_\lambda \KL=0,\\
                                                             &(\Gamma+R) \cdot \nabla_a \KL=0.
                                                           \end{cases}
\end{align*}
In other words, the~flow of internal and active states
minimizes free energy, rendering the variational density
equivalent to the posterior density over external states.\hl{''}
\end{quote}

The first problem in the above quote is that the minimization of the divergence does not follow from the vanishing gradients. On~the contrary, since Equations~(\ref{eq:gradakl0}) and  (\ref{eq:gradlambdakl0}) must hold for all $(s, a, \lambda)$, the~KL divergence $$\KL[q(\Psi|\lambda)||p^*(\Psi|s,a,\lambda)]$$ cannot depend on $(\lambda,a)$; it therefore has no extremum (and thus no minimum) with respect to either of these~coordinates.

The second problem pertains to the identification of the two distributions at a minimum. In~general, if~we try to find the minimum of a KL divergence between a given probability density $p_1(Y)$ and a  family of densities $p_2(Y|\theta)$ parameterized by $\theta$, then the lowest possible value of zero is achieved only if there is a parameter $\theta_1$ such that $p_2(Y|\theta_1)=p_1(Y)$. If~there is no such $\theta_1$, then the minimum value will be larger than zero. So, even if the divergence were minimized, it would not need to be zero. 
More generally, the~divergence $K(s)$ need not be zero for any value of $s$.   

There is therefore no satisfactory reason given why the variational density $q(\Psi|\lambda)$ and the posterior density $p^*(\Psi|s,a,\lambda)$ should be equal or have low KL divergence. In~fact they need not be. {\hl{(}Note that, since any $q(\Psi|\lambda)$ that doesn't depend on $(s,a)$ is an element of the set of those that do, \cref{obs:pqcounter} remains true for the case where we allow this dependence. In~that case, the~Free Energy Lemma holds because we can set $q(\Psi|s,a,\lambda):=p^*(\Psi|s,a,\lambda)$, and~thus a $q$ exists for which the densities are actually equal. However, the~claim here is that for \emph{every} $q$ that obeys the conditions in \cref{obs:pqcounter} we must have equality.\hl{)}}

\begin{obs}
\label{obs:pqcounter}
Given a random dynamical system obeying Equation (\ref{eq:langevin}), ergodicity, \cref{con:flow} and \cref{con:fact}. Then if, additionally, 
\begin{description}
\item[(i)] Equations~(\ref{eq:fa}) and  (\ref{eq:flambda})  hold and the Free Energy Lemma holds i.e.,  there exists a probability density $q(\Psi|\lambda)$ such that  Equations~(\ref{eq:fafe}) and  (\ref{eq:flambdafe})  hold, or~
\item[(ii)] Equations~(\ref{eq:fa3}) and  (\ref{eq:flambda3})  hold and there exists $q(\Psi|\lambda)$ such that Equations~(\ref{eq:fafe4}) and  (\ref{eq:flambdafe4})  hold, or~\item[(iii)] Equations~(\ref{eq:fa2}) and  (\ref{eq:flambda2})  hold and there exists $q(\Psi|\lambda)$ such that Equations~(\ref{eq:fafe3}) and  (\ref{eq:flambdafe3})  hold,
\end{description}
then there is no $c\geq 0$ for which it can be guaranteed that
\begin{align}
    \KL[q(\Psi|\lambda)||p^*(\Psi|s,a,\lambda)]<c.
\end{align}
In particular, it does not follow from these conditions that
\begin{align}
    \label{eq:pqequal}
    q(\Psi|\lambda)=p^*(\Psi|s,a,\lambda).
\end{align}
\end{obs}
\begin{proof}
By example; see \cref{app:pqcounter}. To~show that the implication does not generally hold for given system and densities $q(\Psi|\lambda)$ that obey  Equations (\ref{eq:fa}), (\ref{eq:flambda}), (\ref{eq:fafe}) and \linebreak (\ref{eq:flambdafe}),  Equations (\ref{eq:fa3}), (\ref{eq:flambda3}), (\ref{eq:fafe4}) and (\ref{eq:flambdafe4}), or~ Equations (\ref{eq:fa2}), (\ref{eq:flambda2}), (\ref{eq:fafe3}) and (\ref{eq:flambdafe3}) we only have to consider a system that obeys all three pairs of equations,  Equations (\ref{eq:fa}) and (\ref{eq:flambda}),  Equations (\ref{eq:fa3}) and (\ref{eq:flambda3}), and~ Equations (\ref{eq:fa3}) and (\ref{eq:flambda3}),  and~for which suitable $q(\Psi|\lambda)$ exist. For~this system we then need to show that the $q(\Psi|\lambda)$ that obey  Equations (\ref{eq:fafe}) and (\ref{eq:flambdafe}) are not necessarily equal (or similar) to $p^*(\Psi|s,a,\lambda)$.

We use a variant of the model used in \cref{app:step2counter} as such a counterexample. This system obeys all three of  Equations (\ref{eq:fa}) and (\ref{eq:flambda}),  Equations (\ref{eq:fa3}) and (\ref{eq:flambda3}), and \linebreak Equations (\ref{eq:fa3}) and (\ref{eq:flambda3}) and the nullspace of the associated $\Gamma+R$ is trivial. We identify a set of possible $q(\Psi|\lambda)$ satisfying  Equations (\ref{eq:fafe}) and (\ref{eq:flambdafe}) which implies that the \emph{gradients of the KL divergence between those $q(\Psi|\lambda)$ and $p^*(\Psi|s,a,\lambda)$} vanish i.e.,  Equations \linebreak (\ref{eq:gradakl0}) and (\ref{eq:gradlambdakl0}) hold. We then demonstrate that for the $q(\Psi|\lambda)$ in this set the \emph{value of the KL divergence to $p^*(\Psi|s,a,\lambda)$} can be arbitrarily large. 
\end{proof}

\section{Interpretation}
\label{itm:interpretation}
Finally, we turn our attention to the interpretation in terms of Bayesian inference, i.e., \cref{itm:interpretation}. 
We again quote directly from \citep{friston_life_2013}: 
\begin{quote}
  Because (by Gibbs inequality) this divergence [ $\KL[q(\psi|\lambda)||p^*(\psi|s,a,\lambda)]$] cannot be less
than zero, the~internal flow will appear to have minimized
the divergence between the variational and posterior density. In~other words, the~internal states will appear to have solved
the problem of Bayesian inference by encoding posterior
beliefs about hidden (external) states, under~a generative
model provided by the Gibbs energy.
\end{quote}

We have shown that, in~general, there is no suitable variational density that is only parameterized by the internal coordinate $\lambda$. We then showed that, even if there is a suitable variational density (including those parameterized by all of $(s,a,\lambda)$), it can be arbitrarily different from the posterior density.
Since the arguments for the internal flow appearing to minimize the divergence between variational and posterior density are therefore incorrect, there is no reason why the internal states should appear to have solved the problem of Bayesian~inference.

As mentioned in \cref{itm:exq}, some newer works (e.g.,  \citep{friston_free_2019,parr_markov_2019}) formulate a different \linebreak Free Energy principle, where the variational density of beliefs is parameterised \linebreak not by the internal coordinates $\lambda$ but by $\bar{\lambda}(s,a) = \argmax_\lambda p^*(\lambda|s,a)$, the~most likely value of the internal coordinates given the sensory and active ones. In~this case, \cref{obs:felcounter,obs:pqcounter} do not apply. However, the~new parameters $\bar{\lambda}(s,a)$ are strictly a function of the sensory and active coordinates. This means we have a Markov chain {\hl{(}with capitalisations indicating random variables associated to the corresponding lower case coordinates (or functions of coordinates)\hl{)}}  
 $\Lambda \rightarrow (S,A) \rightarrow \bar{\Lambda}$ and, by~the data processing inequality~\citep{cover_elements_2006}, the~mutual information between the both sensory and active coordinates and the belief parameter $\bar{\lambda}$ upper bounds that between the internal coordinates and the belief parameter. It is therefore not clear to what extent the internal coordinates $\lambda$, rather than the active and sensory coordinates $(s,a)$ themselves, can be said to be encoding beliefs about the external coordinates.  
Note also that, on~any given trajectory, unless~the distribution $p^*(\lambda|s,a)$ is sufficiently peaked and unimodal, the~internal coordinates are not guaranteed to spend most of their time close to their most likely conditional value, and~(by definition if \cref{con:fact} holds), they will not be better predictors of the external coordinates than those in the Markov~blanket.

Generally, $\lambda \neq \bar{\lambda}$ and $\bar{\lambda}$ is the solution to an optimization problem that is assumed to be solved in these later works. Using this optimized variable to parameterise beliefs is therefore a considerable departure from~\cite{friston_life_2013}. Contrary to the impression created by the way it is referenced in \citep{friston_free_2019,parr_markov_2019}, the~older theory in \citep{friston_life_2013} should be clearly distinguished from the newer ones in these more recent~papers.

\crefalias{section}{section}
\section{Consequences for \citep{friston_cognitive_2014}}
\label{sec:ieeepaper}

\hl{The reference} \citep{friston_cognitive_2014} {argues} for the same interpretation as \citep{friston_life_2013} but there are some differences in the~argument. 

The differences are the~following:
\begin{itemize}
\item In \citep{friston_cognitive_2014},  Equation (\ref{eq:langevin}) is formulated for ``generalized states,'' which we refer to here as generalized coordinates. This means that the variable $x$ is replaced by a multidimensional variable denoted $\tx=(x,x',x'',...)$.  
\item The Markov blanket structure is not explicitly defined via  Equation (\ref{eq:mblanket}). Formally, it is introduced directly (see \citep{friston_cognitive_2014} Equation (10)) in a less general form corresponding to  Equations (\ref{eq:fa}) and (\ref{eq:flambda}). {\hl{(}At the same time \citep{friston_life_2013} is referenced in connection to the Markov blanket so there seems to be no intention to replace the original definition with the stronger one.\hl{)}}  
Therefore, our observations concerning \cref{itm:rwfullgrad,itm:rwpartgrad,itm:exq} are not directly relevant to this paper.

\item The internal coordinate $\lambda$ is renamed to $r$ and the role of matrix $R$ is played by the matrix $-Q$.
\item The proof of the Free Energy Lemma given in \citep{friston_cognitive_2014} is different. It (implicitly) suggests to set the variational density equal to the ergodic conditional posterior.
\item The proof of the Free Energy Lemma no longer contains the proposition that the gradient of the KL divergence of the variational density and the ergodic conditional density vanish i.e., \cref{itm:gradvanish}.
\item The proof also no longer contains the claim that the vanishing gradients of the KL divergence of the variational density and the ergodic conditional density imply equality of those densities i.e., \cref{itm:pqequal} is not~present.

\end{itemize}

The interpretation in terms of Bayesian inference is unchanged and still relies on the equality of the variational and the ergodic conditional~density.

Since there are no explicit generalized coordinate versions of \cref{itm:rwfullgrad,itm:rwpartgrad,itm:gradvanish,itm:pqequal} in~\citep{friston_cognitive_2014} we do not discuss those steps here. We only disprove the Free Energy Lemma and the claim that when the Free Energy Lemma holds the variational and ergodic conditional density become equal. For~this we present a way to translate the counterexamples used in \cref{obs:felcounter,obs:pqcounter} into counterexamples in generalized coordinates.
The interpretation in terms of Bayesian inference given in \citep{friston_cognitive_2014} is therefore equally as unjustified as the one in~\citep{friston_life_2013}.

For completeness, we first state the generalized coordinate versions of the stochastic differential equation Equation (\ref{eq:langevin})
\begin{align}
\label{eq:langevingc}
\dot{\tx}=f(\tx) + \tomega,
\end{align} 
the less general version of the Markov blanket structure  Equation (\ref{eq:mblanket})
\begin{align}
\label{eq:mblanketgc}
\begin{split}
f_\tpsi(\tpsi,\ts,\ta) &= (\Gamma - Q)_{\tpsi \tpsi} \nabla_\tpsi \ln p^*(\tpsi,\ts,\ta,\tr)\\
f_\ts(\tpsi,\ts,\ta) &= (\Gamma - Q)_{\ts \ts} \nabla_\ts \ln p^*(\tpsi,\ts,\ta,\tr)\\
f_\ta(\ts,\ta,\tr) &= (\Gamma - Q)_{\ta \ta} \nabla_\ta \ln p^*(\tpsi,\ts,\ta,\tr)\\
f_\tr(\ts,\ta,\tr) &= (\Gamma - Q)_{\tr \tr} \nabla_\tr \ln p^*(\tpsi,\ts,\ta,\tr),
\end{split}
\end{align}
the expression of the $\ta$ and $\tr$ components of the vectorfield in terms of the marginalised ergodic density  Equations (\ref{eq:fa}) and (\ref{eq:flambda})
\begin{align}
f_\ta(\ts,\ta,\tr)&= (\Gamma- Q)_{\ta \ta} \cdot \nabla_\ta \ln p^*(\ts,\ta,\tr),  \label{eq:fagc}\\
f_\tr(\ts,\ta,\tr) &= (\Gamma-Q)_{\tr \tr} \cdot \nabla_\tr \ln p^*(\ts,\ta,\tr),  \label{eq:flambdagc}
\end{align}
and  in terms of free energy  Equations (\ref{eq:fafe}) and (\ref{eq:flambdafe}):
\begin{align}
f_\ta(\ts,\ta,\tr) &= (Q-\Gamma)_{\ta \ta} \cdot \nabla_\ta F(\ts,\ta,\tr), \label{eq:fafegc}\\
f_\tr(\ts,\ta,\tr) &= (Q-\Gamma)_{\tr \tr} \cdot \nabla_\tr F(\ts,\ta,\tr)\label{eq:flambdafegc}.         
\end{align}
The Free Energy Lemma then requires that there exists $q(\tPsi|\tr)$ such that the KL divergence between $p^*(\tPsi|\ts,\ta,\tr)$ vanishes. Without~going into further details of the difference between the proof in \citep{friston_cognitive_2014} and that in \citep{friston_life_2013}, we can prove the former wrong by translating the counterexample used for the latter into generalised coordinates. 
\begin{obs}
There is a general way to translate a system in ordinary coordinates into a system of generalised coordinates that corresponds to an infinite number of independent copies of the original system. This means all properties of the original system (e.g., linearity, ergodicity, the~Gaussian and Markovian property of the noise, \cref{con:flow,con:fact}, properties of $\Gamma,R,U$) are preserved during this translation.
\end{obs}
\begin{proof}
By construction, see \cref{app:ieeecounter}.
\end{proof}
This implies that the counterexamples used in proving \cref{obs:felcounter,obs:pqcounter} directly translate to the setting of the generalised coordinates. The~Free Energy Lemma is therefore also wrong for generalised coordinates and the variational density $q(\tPsi|\tr)$ is not ``ensured''~\citep{friston_cognitive_2014} to be equal to the conditional ergodic density $p^*(\tPsi|\ts,\ta,\tr)$.

\section{Conclusions}
We found that the two different Markov blanket conditions proposed in \citep{friston_life_2013,friston_free_2019,parr_markov_2019} are independent from each other.
We then showed that under both of those Markov blanket conditions, among~the six steps contained in the argument in \citep{friston_life_2013}, three do not hold independently from each other. We also showed that fixing the second of those steps \linebreak  (\cref{itm:rwpartgrad}) does not provide an valid alternative. The~line of reasoning of \citep{friston_life_2013} therefore does not support its claim that the internal coordinates of a Markov blanket ``appear to have solved the problem of Bayesian inference by encoding posterior beliefs about hidden (external) [coordinates], ...''. We have also shown that using generalised coordinates as in \citep{friston_cognitive_2014} does not remedy the situation.
Additionally, we identified a technical error in~\citep{friston_free_2019} and an interpretational issue resulting from possibly too strong assumptions (both \linebreak  \cref{con:flow,con:sol}) in \citep{parr_markov_2019}. 
We also  {highlighted} that the latter publications both argue that it is the most likely internal coordinates given sensory and active coordinates that encode posterior beliefs about external states instead of the internal coordinates themselves.  {The resulting free energy principle and lemma are therefore} a different proposal{. This is} not subject to our technical~critique.

\vspace{6pt} 



\section*{Contributions}
Conceptualization, M.B., F.A.P. and R.K.; Formal analysis, M.B. and F.A.P.; Funding acquisition, F.A.P. and R.K.; Methodology, M.B., F.A.P. and R.K.; Visualization, M.B. and F.A.P.; Writing---original draft, M.B. and F.A.P.; Writing---review \& editing, M.B., F.A.P. and R.K.
All authors have read and agreed to the published version of the~manuscript.

\section*{Funding}
The work by Martin Biehl and Ryota Kanai on this publication was made possible through the support of a grant from Templeton World Charity Foundation, Inc. The~opinions expressed in this publication are those of the authors and do not necessarily reflect the views of Templeton World Charity Foundation, Inc. \hl{Martin Biehl and Ryota Kanai are also funded by the Japan Science and Technology Agency (JST) CREST project.} Felix A. Pollock acknowledges support from the Monash University Network of Excellence for Consciousness and Complexity in the Conscious~Brain.




\section*{Acknowledgments}
All authors are grateful to Karl Friston and Thomas Parr for constructive feedback on an earlier version of this work. We also want to thank Danijar Hafner for pointing us to~\citep{ma_complete_2015}.
Martin Biehl wants to thank Yen Yu for helpful discussions on generalized coordinates.

\section*{Conflicts of interest}
The authors declare no conflict of interest. The~funders had no role in the design of the study; in the collection, analyses, or~interpretation of data; in the writing of the manuscript, or~in the decision to publish the~results.


\begin{appendices}
\crefalias{section}{appsec}



\section{Counterexamples for \cref{obs:condcounter}}
\label{app:conditioncounter}
Consider a four dimensional linear system obeying  Equation (\ref{eq:langevin}) for which there are coordinates $x=(\psi,s,a,\lambda)$ with $n_\psi=n_s=n_a=n_\lambda=1$ and
\begin{align}
  f(x)=M x,
\end{align}
with the parameterisation
\begin{align}
\label{eq:mblanketM}
  M= \left(
\begin{array}{cccc}
 -1 & m_{1} & m_{2} & m_{3} \\
 m_{2} & -1 & m_{2} & m_{3} \\
 m_{3} & m_{2} & -1 & m_{2} \\
 m_{3} & m_{2} & m_{1} & -1
\end{array}
\right).
\end{align}
From  Equation (\ref{eq:con1mat}), it is clear that the system obeys \cref{con:flow} if $m_3=0$. In~this case, taking $\Gamma$ to be the identity matrix, it is possible to show that
\begin{gather}
    U_{\psi \lambda} = -\frac{m_2(m_1-m_2+2)(m_2^3+m_1^2m_2-2m_1m_2^2-4m_1-2)}{(m_1^2+m_2^2-4m_2+4)(m_1^2+5m_2^2-4m_1m_2+4m_2+4)}.
\end{gather}
For fixed, finite $m_2$, this is zero only for a few discrete values of $m_1$, such as $m_1 = m_2-2$; that it is generically non-zero  proves  Equation (\ref{eq:flownotimplyfact}). As~a concrete example, the~following
\begin{align}
  M=\left(
\begin{array}{cccc}
 -1 & -\nicefrac{2}{3} & -\nicefrac{2}{3} & 0 \\
 -\nicefrac{2}{3} & -1 & -\nicefrac{2}{3} & 0 \\
 0 & -\nicefrac{2}{3} & -1 & -\nicefrac{2}{3} \\
 0 & -\nicefrac{2}{3} & -\nicefrac{2}{3} & -1 \\
\end{array}
\right),
\end{align}
has
\begin{align}
    R = \left(
\begin{array}{cccc}
 0 & -\nicefrac{1}{8} & \nicefrac{3}{8} & 0 \\
 \nicefrac{1}{8} & 0 & 0 & -\nicefrac{3}{8} \\
 -\nicefrac{3}{8} & 0 & 0 & \nicefrac{1}{8} \\
 0 & \nicefrac{3}{8} & -\nicefrac{1}{8} & 0 \\
\end{array}
\right)
\end{align}
and (full rank and hence ergodic)
\begin{align}
  U=\left(
\begin{array}{cccc}
\nicefrac{236}{255} & \nicefrac{127}{255} & -\nicefrac{31}{85} & -\nicefrac{12}{85} \\
 \nicefrac{127}{255} & \nicefrac{274}{255} & \nicefrac{206}{255} & \nicefrac{31}{85} \\
 \nicefrac{31}{85} & \nicefrac{206}{255} & \nicefrac{274}{255} & \nicefrac{127}{255} \\
 -\nicefrac{12}{85} & \nicefrac{31}{85} & \nicefrac{127}{255} & \nicefrac{236}{255} \\
\end{array}
\right),
\end{align}
and hence ergodic density
\begin{align}
\label{eq:pstar1}
 p^*(\psi,s,a,\lambda)=&\sqrt{\frac{28}{2295 \pi ^4}}\, \exp\left[-\frac{1}{255} \left(137 (a^2+s^2)+118 (\psi^2+\lambda ^2)\right.\right. \nonumber \\
 &\qquad\left.\left.+127 (\psi s+a \lambda) +93 (\psi a+s \lambda)+206 a s-36   \psi \lambda \vphantom{137 (a^2+s^2)+118 (\psi^2+\lambda ^2)}\right)\vphantom{-\frac{1}{255} \left(137 (a^2+s^2)+118 (\psi^2+\lambda ^2)\right.} \right],
\end{align}
which does not conditionally~factorise.

Taking the same parameterisation as in  Equation (\ref{eq:mblanketM}), and~fixing $m_1 = m_2 = -\nicefrac{1}{2}$, we can search for a non-zero value of $m_3$ that leads to $U_{\psi \lambda} = 0$ (equivalent to \linebreak  \cref{con:fact} through  Equation (\ref{eq:con2mat})). We find such a value in the real root $c\simeq -0.08$ of the quintic equation $8c^5 -4c^4-6c^3+31c^2+40c+3=0$. That is, with~\begin{align}
  M=\left(
\begin{array}{cccc}
 -1 & -\nicefrac{1}{2} & -\nicefrac{1}{2} & c \\
 -\nicefrac{1}{2} & -1 & -\nicefrac{1}{2} & c \\
 c & -\nicefrac{1}{2} & -1 & -\nicefrac{1}{2} \\
 c & -\nicefrac{1}{2} & -\nicefrac{1}{2} & -1 \\
\end{array}
\right),
\end{align}
which does not satisfy \cref{con:flow}, we have
\begin{align}
   R =  \left(
\begin{array}{cccc}
 0 & -0.06\dots & \hphantom{-}0.22\dots & 0 \\
 \hphantom{-}0.06\dots & 0 & 0 & -0.22\dots \\
 -0.22\dots & 0 & 0 & \hphantom{-}0.06\dots \\
 0 & \hphantom{-}0.22\dots & -0.06\dots & 0 \\
\end{array}
\right),
\end{align}
and
\begin{align}
  U = \left(
\begin{array}{cccc}
0.96\dots & 0.43\dots & 0.30\dots & 0 \\
  0.43\dots & 1.03\dots & 0.58\dots & 0.30\dots \\
 0.30\dots & 0.58\dots & 1.03\dots & 0.43\dots \\
 0 &  0.30\dots& 0.43\dots & 0.96\dots \\
\end{array}
\right),
\end{align}
which has non-zero determinant (i.e., the~dynamics is ergodic) and an ergodic density satisfying \cref{con:fact}. This proves  Equation (\ref{eq:factnotimplyflow}).

\section{Counterexample for \cref{itm:rwpartgrad}}
\label{app:step2counter}

Here, we consider a linear system, as~in the previous appendix. We again assume $\Gamma$ equal to the identity matrix and choose a force matrix of the form
\begin{align}
  M=  \left(
\begin{array}{cccc}
 -1 & -\frac{1}{2 \sqrt{2}} & -\frac{1}{2 \sqrt{2}} & 0 \\
 -\frac{1}{2 \sqrt{2}} & -1 & -\frac{1}{16} & 0 \\
 0 & \frac{1}{16} & -1 & -\frac{1}{2 \sqrt{2}} \\
 0 & \frac{1}{2 \sqrt{2}} & -\frac{1}{2 \sqrt{2}} & -1 \\
\end{array}
\right)
\end{align}
which explicitly satisfies \cref{con:flow} and has full rank such that the system is ergodic. Using  Equation (\ref{eq:Rdef}) this leads to
\begin{align}
U=\left(
\begin{array}{cccc}
 \frac{1023}{1057} & \frac{260 \sqrt{2}}{1057} & \frac{136 \sqrt{2}}{1057} & 0 \\
 \frac{260 \sqrt{2}}{1057} & \frac{1091}{1057} & 0 & -\frac{136 \sqrt{2}}{1057} \\
 \frac{136 \sqrt{2}}{1057} & 0 & \frac{1091}{1057} & \frac{260 \sqrt{2}}{1057} \\
 0 & -\frac{136 \sqrt{2}}{1057} & \frac{260 \sqrt{2}}{1057} & \frac{1023}{1057} \\
\end{array}
\right)
\end{align}
which shows that this system also satisfies \cref{con:fact} since $U_{\psi \lambda}=U_{\lambda \psi}=0$. We also find
\begin{align}
    R=
\left(
\begin{array}{cccc}
 0 & -\frac{17}{1786 \sqrt{2}} & \frac{479}{1786 \sqrt{2}} & -\frac{62}{893} \\
 \frac{17}{1786 \sqrt{2}} & 0 & -\frac{65}{14288} & \frac{479}{1786 \sqrt{2}} \\
 -\frac{479}{1786 \sqrt{2}} & \frac{65}{14288} & 0 & \frac{17}{1786 \sqrt{2}} \\
 \frac{62}{893} & -\frac{479}{1786 \sqrt{2}} & -\frac{17}{1786 \sqrt{2}} & 0 \\
\end{array}
\right),
\end{align}
which shows that all entries or $R$ that can be non-zero for an anti-symmetric matrix are non-zero.
For the marginal ergodic density we find
\begin{align}
\begin{split}
p^*(s,a,\lambda)=\frac{239}{16 \sqrt{2415} \pi ^{3/2}}  \exp&\left[-\frac{69 a^2}{140}-\frac{37 a s}{70 \sqrt{2}}+\frac{1}{35} \sqrt{2} 2 a \psi \right.\\
&\left.\phantom{\left[\vphantom{-\frac{4867 s^2}{9660}}\right.}-\frac{4867 s^2}{9660}+\frac{74 s \psi }{2415}-\frac{8429 \psi ^2}{19320}\right]
\end{split}
\end{align}
The difference between the right hand sides of  Equations (\ref{eq:fa2}) and (\ref{eq:fa}) is
\begin{align}
  R_{a s} \nabla_s \ln p^*(s,a,\lambda)+ R_{a \lambda} \nabla_\lambda \ln p^*(s,a,\lambda)
  =\frac{37 a+69 \sqrt{2} \lambda -2563 s}{4830}
  \neq0,
\end{align}
which shows that  Equation (\ref{eq:fa}) is wrong in this example and therefore not generally equivalent to  Equation (\ref{eq:fa2}). Similarly, computing the difference between the right hand sides of  Equations (\ref{eq:flambda2}) and (\ref{eq:flambda}), one finds
\begin{align}
    R_{\lambda s} \nabla_s \ln p^*(s,a,\lambda)+ R_{\lambda a} \nabla_a \ln p^*(s,a,\lambda)
    =\frac{2 a-\sqrt{2} \lambda +27 s}{70 \sqrt{2}}
    \neq0,
\end{align}
and hence  Equation (\ref{eq:flambda}) is also incorrect in~general.

Performing the same comparison for the difference between the general expression in  Equations (\ref{eq:fa2}) and (\ref{eq:flambda2}) and the expressions taken from~\citep{friston_free_2019}, one finds
\begin{align}
  R_{a_i s} \nabla_s \ln p^*(s,a,\lambda)=\frac{73 \left(296 a+552 \sqrt{2} \lambda -8429 s\right)}{1,154,370}\neq 0
\end{align}
for the difference between the right hand sides of  Equations (\ref{eq:fa2}) and (\ref{eq:fa3}), and~\begin{align}
 R_{\lambda s} \nabla_s \ln p^*(s,a,\lambda)=-\frac{53 \left(296 a+552 \sqrt{2} \lambda -8429 s\right)}{1,154,370 \sqrt{2}}\neq 0,
\end{align}
for the difference between the right hand sides of  Equations (\ref{eq:flambda2}) and (\ref{eq:flambda3}). Therefore,  Equations (\ref{eq:fa3}) and (\ref{eq:flambda3}) are also incorrect in general, even when \cref{con:flow} and \cref{con:fact} both~hold.

\section{Counterexamples for \cref{itm:exq}}
\label{app:felcounter}
We saw in \cref{app:step2counter} that  Equations (\ref{eq:fa}) and (\ref{eq:flambda}) are not generally equivalent to  Equation (\ref{eq:25}), even when \cref{con:flow} and \cref{con:fact} hold simultaneously. We now show that if we instead use  Equations (\ref{eq:fa2}) and (\ref{eq:flambda2}), which are generally equivalent to  \linebreak Equation (\ref{eq:25}), the~Free Energy Lemma does not hold in~general.

The original Free Energy Lemma requires that (see  Equations (\ref{eq:fafe2}) and (\ref{eq:flambdafe2}))
\begin{align}
 (\Gamma+R)_{aa} \cdot \nabla_a \KL[q(\Psi|\lambda)||p^*(\Psi|s,a,\lambda)]&=0\label{eq:fela0}\\
 (\Gamma+R)_{\lambda \lambda} \cdot \nabla_\lambda \KL[q(\Psi|\lambda)||p^*(\Psi|s,a,\lambda)]&=0.\label{eq:fellambda0}          
\end{align}
replacing the partial gradient in  Equations (\ref{eq:fafe}) and (\ref{eq:flambdafe}) with the full gradient and including the entire matrix $(\Gamma+R)$ leads to the corresponding requirement for the more general case:
\begin{align}
\left( R_{a s} \cdot \nabla_s + (\Gamma_{a a}+ R_{a a}) \cdot \nabla_a + R_{a \lambda} \cdot \nabla_\lambda \right) \KL[q(\Psi|\lambda)||p^*(\Psi|s,a,\lambda)]&=0  \label{eq:fela02}\\
\left( R_{\lambda s} \cdot \nabla_s + R_{\lambda a} \cdot \nabla_a + (\Gamma_{\lambda \lambda}+R_{\lambda \lambda}) \cdot \nabla_\lambda \right) \KL[q(\Psi|\lambda)||p^*(\Psi|s,a,\lambda)]&=0.  \label{eq:fellambda02}
\end{align}
Similarly, the~version based on the equations taken from~\citep{friston_free_2019} implies
\begin{align}
    \left( (\Gamma_{a a}+ R_{a a}) \cdot \nabla_a + R_{a \lambda} \cdot \nabla_\lambda \right) \KL[q(\Psi|\lambda)||p^*(\Psi|s,a,\lambda)]&=0  \label{eq:fela03}\\
\left(  R_{\lambda a} \cdot \nabla_a + (\Gamma_{\lambda \lambda}+R_{\lambda \lambda}) \cdot \nabla_\lambda \right) \KL[q(\Psi|\lambda)||p^*(\Psi|s,a,\lambda)]&=0.  \label{eq:fellambda03}
\end{align}

Using the rules of Gaussian integration, we can write the logarithm of the conditional ergodic density as
\begin{align}
 \ln p^*(\psi|s,a,\lambda) = -\frac{1}{2}|U_{\psi\psi}^{\frac{1}{2}}\psi + U_{\psi\psi}^{-\frac{1}{2}}U_{\psi s} s + U_{\psi\psi}^{-\frac{1}{2}}U_{\psi a} a + U_{\psi\psi}^{-\frac{1}{2}}U_{\psi \lambda} \lambda|^2 + C, \label{eq:condexponent}
\end{align}
with $C$ a constant (and remembering each of $\psi$, $s$, $a$ and $\lambda$ is a vector of coordinates in general). We can then expand out the derivatives of the KL divergence to express them in terms of the coordinates:
\begin{align}
    \nabla_s \KL[q(\Psi|\lambda)||p^*(\psi|s,a,\lambda)] =& -\int {\rm d}\psi\, q(\psi|\lambda)\nabla_s\ln p^*(\psi|s,a,\lambda) \nonumber \\
    =&  U_{s\psi}U_{\psi\psi}^{-1}\big(U_{\psi s } s + U_{\psi a } a \nonumber \\
    &\qquad\qquad+ U_{\psi \lambda } \lambda + U_{\psi \psi } \langle \psi \rangle_{q(\Psi|\lambda)}\big), \label{eq:gradKLs}\\
    \nabla_a \KL[q(\Psi|\lambda)||p^*(\psi|s,a,\lambda)] =& -\int {\rm d}\psi\, q(\psi|\lambda)\nabla_a\ln p^*(\psi|s,a,\lambda)\nonumber \\
    =&  U_{a\psi}U_{\psi\psi}^{-1}\big(U_{\psi s } s + U_{\psi a } a \nonumber \\
    &\qquad\qquad+ U_{\psi \lambda } \lambda + U_{\psi \psi } \langle \psi \rangle_{q(\Psi|\lambda)}\big), \label{eq:gradKLa}\\
    \nabla_\lambda \KL[q(\Psi|\lambda)||p^*(\psi|s,a,\lambda)] =&\int {\rm d}\psi\,\Big[\left(\ln q(\psi|\lambda) - \ln p^*(\psi|s,a,\lambda) + 1\right)\nabla_\lambda q(\psi|\lambda) \nonumber \\
  &\qquad \quad- q(\psi|\lambda)\nabla_\lambda\ln p^*(\psi|s,a,\lambda)\Big]\nonumber \\
    =&  U_{\lambda\psi}U_{\psi\psi}^{-1}\big(U_{\psi s } s + U_{\psi a } a \nonumber \\
    &\qquad\qquad+ U_{\psi \lambda } \lambda + U_{\psi \psi } \langle \psi \rangle_{q(\Psi|\lambda)}\big) \nonumber \\
    &+ \nabla_\lambda\langle \psi \rangle_{q(\Psi|\lambda)}\big(U_{\psi s } s + U_{\psi a } a + U_{\psi \lambda } \lambda \big) \nonumber \\
    & + \nabla_\lambda\left(\langle \psi^T U_{\psi\psi} \psi \rangle_{q(\Psi|\lambda)}-H[q(\Psi|\lambda)]\right),\label{eq:gradKLlambda}
\end{align}
with $\langle g(\psi)\rangle_{q(\Psi|\lambda)}:=\int{\rm d}\psi\, q(\psi|\lambda) g(\psi)$ and $H$ the Shannon~entropy.

Substituting  Equations (\ref{eq:gradKLa}) and (\ref{eq:gradKLlambda}) into  Equations (\ref{eq:fela0}) and (\ref{eq:fellambda0}) leads to
\begin{align}
    (\Gamma_{aa} + R_{aa})U_{a\psi}U_{\psi\psi}^{-1}\left(U_{\psi s } s + U_{\psi a } a + U_{\psi \lambda } \lambda + U_{\psi \psi } \langle \psi \rangle_{q(\Psi|\lambda)}\right) = 0,
\end{align}
and
\begin{align}
    0=&(\Gamma_{\lambda\lambda} + R_{\lambda\lambda})U_{\lambda\psi}U_{\psi\psi}^{-1}\big(U_{\psi s } s + U_{\psi a } a + U_{\psi \lambda } \lambda + U_{\psi \psi } \langle \psi \rangle_{q(\Psi|\lambda)}\big) \nonumber \\
    &+ (\Gamma_{\lambda\lambda} + R_{\lambda\lambda})\nabla_\lambda\langle \psi \rangle_{q(\Psi|\lambda)}\big(U_{\psi s } s + U_{\psi a } a + U_{\psi \lambda } \lambda \big) \nonumber \\
    & + (\Gamma_{\lambda\lambda} + R_{\lambda\lambda})\nabla_\lambda\left(\langle \psi^T U_{\psi\psi} \psi \rangle_{q(\Psi|\lambda)}-H[q(\Psi|\lambda)]\right).
\end{align}
Since these must hold for all values of the coordinates, they put strong requirements on the $U$ and $R$ matrices. Specifically,
\begin{align}
     (\Gamma_{aa} + R_{aa})U_{a\psi}U_{\psi\psi}^{-1}U_{\psi s} &= 0, \\ (\Gamma_{aa} + R_{aa})U_{a\psi}U_{\psi\psi}^{-1}U_{\psi a} &= 0, \\ (\Gamma_{aa} + R_{aa})U_{a\psi}U_{\psi\psi}^{-1}U_{\psi \lambda} &= 0.
\end{align}
In other words, since $U_{\psi\psi}$ and $\Gamma_{aa}$ must be nonzero for the dynamics to be ergodic, it must be that $U_{\psi a}=0$. {\hl{(}This is equivalent to $p^*(\Psi|s,a,\lambda)=p^*(\Psi|s,\lambda)$. So if \cref{con:fact} also holds we must have $p^*(\Psi|s,a,\lambda)=p^*(\Psi|s)$ in order for there to be a suitable $q(\Psi|\lambda)$.\hl{)}}
Specifically, consider the system specified by the force matrix
\begin{align}
    M=\left(
\begin{array}{cccc}
 -1 & 0 & \frac{1}{2} & 0 \\
 0 & -1 & 0 & 0 \\
 0 & 0 & -1 & 0 \\
 0 & 0 & 0 & -1 \\
\end{array}
\right)
\end{align}
leads to
\begin{align}
    R=\left(
\begin{array}{cccc}
 0 & 0 & -\frac{1}{4} & 0 \\
 0 & 0 & 0 & 0 \\
 \frac{1}{4} & 0 & 0 & 0 \\
 0 & 0 & 0 & 0 \\
\end{array}
\right)
\end{align}
and
\begin{align}
    U=\left(
\begin{array}{cccc}
 \frac{16}{17} & 0 & -\frac{4}{17} & 0 \\
 0 & 1 & 0 & 0 \\
 -\frac{4}{17} & 0 & \frac{18}{17} & 0 \\
 0 & 0 & 0 & 1 \\
\end{array}
\right).
\end{align}
Here $M$ is full rank so the system is ergodic, clearly it also satisfies \cref{con:flow} due to the structure of $M$. Since $R_{as}=R_{a \lambda}=R_{\lambda s}=0$ it obeys  Equations (\ref{eq:fa}) and (\ref{eq:flambda}) and since $Q_{\psi \lambda}=0$ it also obeys \cref{con:fact}. Additionally, we find $U_{\psi a}= -\nicefrac{4}{17}$ which is a~contradiction.

For the more general version, 
substituting  Equations (\ref{eq:gradKLs})--(\ref{eq:gradKLlambda}) into  Equation (\ref{eq:fela02}), one finds
\begin{align}
    0=&\left(\left(R_{as}U_{s\psi} + (R_{aa} + \Gamma_{aa})U_{a\psi} + R_{a\lambda}U_{\lambda\psi} \right)U_{\psi\psi}^{-1} + R_{a\lambda}\nabla_\lambda\langle \psi \rangle_{q(\Psi|\lambda)} \right) U_{\psi s}s  \nonumber \\
    &+ \left(\left(R_{as}U_{s\psi} + (R_{aa} + \Gamma_{aa})U_{a\psi} + R_{a\lambda}U_{\lambda\psi} \right)U_{\psi\psi}^{-1} + R_{a\lambda}\nabla_\lambda\langle \psi \rangle_{q(\Psi|\lambda)} \right) U_{\psi a} a \nonumber \\
    &+ \left(\left(R_{as}U_{s\psi} + (R_{aa} + \Gamma_{aa})U_{a\psi} + R_{a\lambda}U_{\lambda\psi} \right)U_{\psi\psi}^{-1} + R_{a\lambda}\nabla_\lambda\langle \psi \rangle_{q(\Psi|\lambda)} \right) U_{\psi \lambda} \lambda \nonumber \\
    &+ \left(R_{as}U_{s\psi} + (R_{aa} + \Gamma_{aa})U_{a\psi} + R_{a\lambda}U_{\lambda\psi} \right) \langle \psi \rangle_{q(\Psi|\lambda)} \nonumber \\
    &+ R_{a\lambda}\nabla_\lambda\left(\langle \psi^T U_{\psi\psi} \psi \rangle_{q(\Psi|\lambda)}-H[q(\Psi|\lambda)]\right),
\end{align}
which, considering that the coordinates can take any values, implies that
\begin{gather} \label{eq:aimplicationquantity}
    \left(R_{as}U_{s\psi} + (R_{aa} + \Gamma_{aa})U_{a\psi} + R_{a\lambda}U_{\lambda\psi} \right)U_{\psi\psi}^{-1} + R_{a\lambda}\nabla_\lambda\langle \psi \rangle_{q(\Psi|\lambda)}
\end{gather}
lies in a common (left) nullspace of $U_{\psi s}$, $U_{\psi a}$ and $U_{\psi \lambda}$. However, the~existence of such a nontrivial nullspace would imply that the corresponding subspace of $\psi$ coordinates is independent of the $s$, $a$ and $\lambda$ coordinates (to see this, consider marginalising over their complement in  Equation (\ref{eq:condexponent})). In~other words, if~only $\psi$ coordinates that play a nontrivial role in the dynamics are considered, then  Equation (\ref{eq:fela02}) must imply the quantity in  Equation (\ref{eq:aimplicationquantity}) is zero, and~hence that
\begin{gather}
    R_{a\lambda}\nabla_\lambda\langle \psi \rangle_{q(\Psi|\lambda)} = -\left(R_{as}U_{s\psi} + (R_{aa} + \Gamma_{aa})U_{a\psi} + R_{a\lambda}U_{\lambda\psi} \right)U_{\psi\psi}^{-1}. \label{eq:contradictiona}
\end{gather}

However, through a similar procedure, one finds that  Equation (\ref{eq:fellambda02}) is equivalent to
\begin{align}
    0=&\Big(\left(R_{\lambda s}U_{s\psi} + R_{\lambda a}U_{a\psi} + (\Gamma_{\lambda \lambda} + R_{\lambda\lambda})U_{\lambda\psi} \right)U_{\psi\psi}^{-1} \nonumber \\&\qquad+ (\Gamma_{\lambda \lambda} + R_{\lambda\lambda})\nabla_\lambda\langle \psi \rangle_{q(\Psi|\lambda)} \Big) U_{\psi s}s  \nonumber \\
    &+ \Big(\left(R_{\lambda s}U_{s\psi} + R_{\lambda a}U_{a\psi} + (\Gamma_{\lambda \lambda} + R_{\lambda\lambda})U_{\lambda\psi} \right)U_{\psi\psi}^{-1} \nonumber \\&\qquad+ (\Gamma_{\lambda \lambda} + R_{\lambda\lambda})\nabla_\lambda\langle \psi \rangle_{q(\Psi|\lambda)} \Big) U_{\psi a} a \nonumber \\
    &+ \Big(\left(R_{\lambda s}U_{s\psi} + R_{\lambda a}U_{a\psi} + (\Gamma_{\lambda \lambda} + R_{\lambda\lambda})U_{\lambda\psi} \right)U_{\psi\psi}^{-1} \nonumber \\&\qquad+ (\Gamma_{\lambda \lambda} + R_{\lambda\lambda})\nabla_\lambda\langle \psi \rangle_{q(\Psi|\lambda)} \Big) U_{\psi \lambda} \lambda \nonumber \\
    &+ \left(R_{\lambda s}U_{s\psi} + R_{\lambda a}U_{a\psi} + (\Gamma_{\lambda \lambda} + R_{\lambda\lambda})U_{\lambda\psi} \right) \langle \psi \rangle_{q(\Psi|\lambda)} \nonumber \\
    &+ (\Gamma_{\lambda \lambda} + R_{\lambda\lambda})\nabla_\lambda\left(\langle \psi^T U_{\psi\psi} \psi \rangle_{q(\Psi|\lambda)}-H[q(\Psi|\lambda)]\right),
\end{align}
implying that
\begin{gather}
    (\Gamma_{\lambda \lambda} + R_{\lambda\lambda})\nabla_\lambda\langle \psi \rangle_{q(\Psi|\lambda)} = -\left(R_{\lambda s}U_{s\psi} + R_{\lambda a}U_{a\psi} + (\Gamma_{\lambda \lambda} + R_{\lambda\lambda})U_{\lambda\psi} \right)U_{\psi\psi}^{-1}. \label{eq:contradictionlambda}
\end{gather}
Unless $R_{a\lambda}$ and $ (\Gamma_{\lambda \lambda} + R_{\lambda\lambda})$ share a common nullspace, or~the $U$ and $R$ matrices are finely tuned, then  Equations (\ref{eq:contradictiona}) and (\ref{eq:contradictionlambda}) contradict one another. In~this case, there cannot exist a $q(\Psi|\lambda)$ that satisfies both  Equations (\ref{eq:fela02}) and (\ref{eq:fellambda02}), and~hence the modified Free Energy Lemma is invalid in general. In~particular, using the example from \cref{app:step2counter}, if~we solve  Equation (\ref{eq:contradictiona}) for $\nabla_\lambda\langle \psi \rangle_{q(\Psi|\lambda)}$ we find
\begin{align}
   \nabla_\lambda\langle \psi \rangle_{q(\Psi|\lambda)}=-\frac{53}{5},
   \end{align}
and from  Equation (\ref{eq:contradictionlambda}) we get
\begin{align}
    \nabla_\lambda\langle \psi \rangle_{q(\Psi|\lambda)}=\frac{29}{239},
\end{align}
which is a~contradiction.

If we now perform the same procedure for  Equations (\ref{eq:fela03}) and (\ref{eq:fellambda03}), we arrive at the following conditions on the gradient of the variational density:
\begin{align} \label{eq:contradictiona2}
    R_{a\lambda} \nabla_\lambda\langle \psi \rangle_{q(\Psi|\lambda)} = -\left((R_{aa} + \Gamma_{aa})U_{a\psi} + R_{a\lambda}U_{\lambda\psi} \right)U_{\psi\psi}^{-1}.
\end{align}
and
\begin{align}\label{eq:contradictionlambda2}
    R_{a\lambda}\nabla_\lambda\langle \psi \rangle_{q(\Psi|\lambda)} = -\left(R_{\lambda a}U_{a\psi} + (\Gamma_{\lambda \lambda} + R_{\lambda\lambda})U_{\lambda\psi} \right)U_{\psi\psi}^{-1}.
\end{align}
Even when \cref{con:fact} holds and $U_{\psi \lambda}=0$, these will be inconsistent in general. As~a specific counterexample, take the system with force matrix
\begin{align}
    M=\left(
\begin{array}{cccc}
 -1 & 0 & -\frac{1}{2} & 0 \\
 0 & -1 & 0 & 0 \\
 0 & 0 & -1 & \frac{\sqrt{3}}{2} \\
 0 & 0 & 0 & -\frac{1}{2} \\
\end{array}
\right),
\end{align}
with correposponding
\begin{align}
    R=\left(
\begin{array}{cccc}
 0 & 0 & \frac{1}{3} & \frac{1}{3\sqrt{3}} \\
 0 & 0 & 0 & 0 \\
 -\frac{1}{3} & 0 & 0 & -\frac{1}{\sqrt{3}} \\
 -\frac{1}{3\sqrt{3}} & 0 & \frac{1}{\sqrt{3}} & 0 \\
\end{array}
\right),
\end{align}
and
\begin{align}
    U=\left(
\begin{array}{cccc}
 \frac{9}{10} & 0 & \frac{3}{10} & 0 \\
 0 & 1 & 0 & 0 \\
 \frac{3}{10} & 0 & \frac{17}{20} & -\frac{\sqrt{3}}{4} \\
 0 & 0 & -\frac{\sqrt{3}}{4} & \frac{3}{4} \\
\end{array}
\right).
\end{align}
This model is ergodic (full rank $U$), and~it satisfies both \cref{con:flow} and \cref{con:fact}. Moreover, the~forces satisfy  Equations (\ref{eq:fa3}) and (\ref{eq:flambda3}). However, substituting the relevant elements of $U$ and $R$ matrices into  Equation (\ref{eq:contradictiona2}), we find
\begin{align}
   \nabla_\lambda\langle \psi \rangle_{q(\Psi|\lambda)} = \frac{1}{\sqrt{3}}, 
\end{align}
but doing the same for  Equation (\ref{eq:contradictionlambda2}) gives
\begin{align}
    \nabla_\lambda\langle \psi \rangle_{q(\Psi|\lambda)} = \frac{1}{3},
\end{align}
which is a~contradiction.

\section{Counterexample for \cref{itm:pqequal}}
\label{app:pqcounter}
Here we provide an example system for which \cref{con:flow,con:fact} as well as  \cref{itm:rwfullgrad,itm:rwpartgrad,itm:exq,itm:gradvanish} are valid but \cref{itm:pqequal} fails. 
We use a system with
\begin{align}
  f(x)=M x
\end{align}
where
\begin{align}
\label{eq:Mpqcounter}
  M:=\left(
\begin{array}{cccc}
 -1 & \nicefrac{1}{2} & 0 & 0 \\
 \nicefrac{1}{2} & -1 & \nicefrac{1}{2} & 0 \\
 0 & \nicefrac{1}{2} & -1 & \nicefrac{1}{2} \\
 0 & 0 & \nicefrac{1}{2} & -1 \\
\end{array}
\right).
\end{align}
This system is ergodic, satisfies \cref{con:flow} and as we will
will see satisfies  Equations (\ref{eq:fa}) and (\ref{eq:flambda}) as well. Using  Equation (\ref{eq:Rdef}) we find
\begin{align}
  R=\left(
\begin{array}{cccc}
 0 & 0 & 0 & 0 \\
 0 & 0 & 0 & 0 \\
 0 & 0 & 0 & 0 \\
 0 & 0 & 0 & 0 \\
\end{array}
\right)
\end{align}
and from Equation (\ref{eq:Udef})
\begin{align}
    U=-M
\end{align}
which means that \cref{con:fact} is also~satisfied.

This leads to the ergodic density
\begin{align}
p^*(\psi,s,a,\lambda)= & \frac{\sqrt{5}}{16 \pi ^2} e^{-\frac{1}{2} \left(\psi^2-\psi s+s^2-s a+a^2-a \lambda+\lambda^2\right)}
\end{align}
which can be used to check that  Equations (\ref{eq:fa}) and (\ref{eq:flambda}) hold for this example.
The conditional ergodic density is
\begin{align}
\label{eq:pstarmpqcounter}
  p^*(\psi|s,a,\lambda)=p^*(\psi|s)=\frac{1}{\sqrt{2 \pi }}e^{-\frac{1}{2} \left(\psi-\frac{1}{2} s\right)^2}.
\end{align}

If we now define $q(\psi|\lambda)=q(\psi) = \exp(-(\psi-\mu)^2/2)/\sqrt{2\pi}$ as a Gaussian distribution with mean $\mu$ and variance one, we can compute the KL divergence to get:
\begin{align}
  \KL[q(\Psi)||p^*(\Psi|s,a,\lambda)]=K(s)=
  \frac{1}{2} \left( \mu -\frac{1}{2} s\right)^2.
\end{align}

Clearly, for~this choice of $q(\psi|\lambda)$ the gradients with respect to $a$ and $\lambda$ of the KL divergence vanish everywhere (Equations (\ref{eq:gradakl0}) and (\ref{eq:gradlambdakl0}) hold). This also means we can express $f_a,f_\lambda$ in terms of a free energy i.e., the Free Energy Lemma holds for this system. 
However, for~any proposed bound $c\geq 0$ on the KL divergence, there is a value of $s$ for which it is exceeded, whatever the choice of $\mu$. Moreover, we can choose a $\mu$ such that the KL divergence is larger than any given $c$, even when $s=0$.

\section{Translating Systems into Generalized Coordinates~Systems}
\label{app:ieeecounter}

We show how to get a generalized coordinate system from a finite dimensional system. By~definition the generalized coordinates are infinite dimensional. For~all $n \in \mathbb{N}$ and a coordinate $x$ they also include the $n$-th time derivative of $x$.

Assume as given an ergodic, linear, random dynamical system described by
\begin{align}
    \dot{x}=M x +\omega
\end{align}
where $x=(x_1,...,x_k)$ is a $k$-dimensional vector, $M$ is a $k \times k$ real valued matrix, and~$\dot{x}:=\ddt x$.
We can look at the second time derivative of the state by differentiating both sides:
\begin{align}
    \ddt \dot{x}&=\ddt (M x + \omega) \\
    \ddot{x}&=M \dot{x} + \dot{\omega}
\end{align}
Similarly for the third time derivative:
\begin{align}
    \ddt \ddot{x}&=\ddt (M \dot{x} + \dot{\omega}) \\
    &=M \ddot{x} + \ddot{\omega}
\end{align}
Similarly for all higher derivatives:
\begin{align}
    \frac{d^n}{d t^n}x&=M  \frac{d^{n-1}}{d t^{n-1}}x+ \frac{d^n}{d t^n}\omega.
\end{align}
Now define the generalized coordinates $\tx=(x,x',x'',...)$ as
\begin{align}
    x &= x\\
    x' &= \ddt x\\
    x'' &= \frac{d^2}{d t^2}x\\
    &\vdots\\
    x^{(n)}&=\frac{d^n}{d t^n}x\\
    &\vdots
\end{align}
Define also
\begin{align} \label{eq:noisederivs}
    \tilde{\omega}:=(\omega,\ddt \omega,\frac{d^2}{d t^2} \omega,...,\frac{d^n}{d t^n} \omega,...).
\end{align}
Without further clarification, the~derivatives of $\omega$ are not well defined when the latter is a Gaussian white noise process, as~explicitly assumed in writing the vector field $f(x)$ in terms of the ergodic density~\citep{ao_potential_2004, kwon_structure_2005, kwon_nonequilibrium_2011}. As~discussed in \citep{vanKampen_stochastic_1981}, delta-correlated Markovian noise is always a limiting approximation of noise with a finite correlation time. Meaningfully taking the derivatives requires first choosing a functional form for the (co)variance whose limit is a delta function {\hl{(}Another, more direct approach would be in terms of generalized functions, but~here too additional information is required to specify the derivatives~\citep{oberguggenberger_generalized_1995}.\hl{)}} However, different choices can lead to vastly different central moments of the generalized noise distribution, including those that vanish or diverge at all orders. In~the former case, the~process in terms of generalized coordinates may not be ergodic~\citep{cornfeld_ergodic_1982}; in the latter case, the~process is not well defined. In~general, it is not clear that Equation \eqref{eq:25} holds in the non-Markovian case, since the standard derivations in~\citep{ao_potential_2004,kwon_structure_2005} and related works rely on delta-correlated~noise. 

Here, we can therefore assume that the noise is such that the derivatives in Equation~\eqref{eq:noisederivs} can be treated as Markov and Gaussian. We also assume that $\frac{d^n}{d t^n} \omega$ is independently and identically distributed to $\frac{d^{n-1}}{d t^{n-1}} \omega$ for all $n$. 
Finally, we can then define the (infinite) matrix $\bM$ as the block diagonal matrix with all blocks equal to $M$:
\begin{align}
    \bM := \begin{pmatrix} 
    M & 0 & \dots & \\
    0 & M  &  &\\
    \vdots &  & \ddots & \\
    \end{pmatrix}
\end{align}
The time derivative of $\omega$ is independent of $\omega$, as~the changes are independent of the value of $\omega$. So we actually get an infinite number of independent and identically distributed systems.
Using these definitions we have:
\begin{align}
   \dot{\tx} = \bM \tx + \tilde{\omega}. 
\end{align}
These equations describe a random dynamical system composed of an infinite number of independent linear random dynamical systems, all governed by the same matrix $M$ and driven by independently and identically distributed noise. Since the first of these systems (for the variables $x$) is ergodic by assumption, all of the subsystems are also ergodic and, therefore, the~whole system is ergodic with the ergodic density equal to a product of the original ergodic density:
\begin{align}
    \bar{p}^*(\tx)=p^*(x) p^*(x') p^*(x'') \cdots p^*(x^{(n)}) \cdots. 
\end{align}
Additionally, if~$M$ is such that
\begin{align}
\label{eq:mblanketM2}
\begin{split}
M_\psi \cdot (\psi,s,a,r)^\top = f_\psi(\psi,s,a) &= (\Gamma - Q)_{\psi \psi} \nabla_\psi \ln p^*(\psi,s,a,r)\\
M_s \cdot (\psi,s,a,r)^\top = f_s(\psi,s,a) &= (\Gamma - Q)_{s s} \nabla_s \ln p^*(\psi,s,a,r)\\
M_a \cdot (\psi,s,a,r)^\top = f_a(s,a,r) &= (\Gamma - Q)_{a a} \nabla_a \ln p^*(\psi,s,a,r)\\
M_r \cdot (\psi,s,a,r)^\top = f_r(s,a,r) &= (\Gamma - Q)_{r r} \nabla_r \ln p^*(\psi,s,a,r),
\end{split}
\end{align}
(which is the case for the $M$ in the counterexample to \cref{itm:pqequal}) then for
\begin{align}
    \begin{split}
        (x_1,x_2,x_3,x_4):&=(\psi,s,a,r)\\
        (x'_1,x'_2,x'_3,x'_4):&=(\psi',s',a',r')\\
        (x''_1,x''_2,x''_3,x''_4):&=(\psi'',s'',a'',r'')\\
        &\vdots\\
        (x^{(n)}_1,x^{(n)}_2,x^{(n)}_3,x^{(n)}_4):&=(\psi^{(n)},s^{(n)},a^{(n)},r^{(n)})\\
        &\vdots,
    \end{split}
\end{align}
\begin{align}
    \bQ := \begin{pmatrix} 
    Q & 0 & \dots & \\
    0 & Q  &  &\\
    \vdots &  & \ddots & \\
    \end{pmatrix},
\end{align}
\begin{align}
    \bGamma := \begin{pmatrix} 
    \Gamma & 0 & \dots & \\
    0 & \Gamma  &  &\\
    \vdots &  & \ddots & \\
    \end{pmatrix},
\end{align}
and using Equation (\ref{eq:Udef}) and that the inverse of a block diagonal matrix is block diagonal
\begin{align}
    \bar{U} := \begin{pmatrix} 
    U & 0 & \dots & \\
    0 & U  &  &\\
    \vdots &  & \ddots & \\
    \end{pmatrix},
\end{align}
we also have:
\begin{align}
\begin{split}
\bM_\tpsi \cdot (\tpsi,\ts,\ta,\tr)^\top = f_\tpsi(\tpsi,\ts,\ta) &= (\bGamma - \bQ)_{\tpsi \tpsi} \nabla_\tpsi \ln p^*(\tpsi,\ts,\ta,\tr)\\
\bM_\ts \cdot (\tpsi,\ts,\ta,\tr)^\top = f_\ts(\tpsi,\ts,\ta) &= (\bGamma - \bQ)_{\ts \ts} \nabla_\ts \ln p^*(\tpsi,\ts,\ta,\tr)\\
\bM_\ta \cdot (\tpsi,\ts,\ta,\tr)^\top = f_\ta(\ts,\ta,\tr) &= (\bGamma - \bQ)_{\ta \ta} \nabla_\ta \ln p^*(\tpsi,\ts,\ta,\tr)\\
\bM_\tr \cdot (\tpsi,\ts,\ta,\tr)^\top = f_\tr(\ts,\ta,\tr) &= (\bGamma - \bQ)_{\tr \tr} \nabla_\tr \ln p^*(\tpsi,\ts,\ta,\tr).
\end{split}
\end{align}

The ergodic density of such a system is a product of ergodic densities of the original system Equation (\ref{eq:pstarmpqcounter}):
\begin{align}
    \bp^*(\tpsi,\ts,\ta,\tr)=p^*(\psi,s,a,r) p^*(\psi',s',a',r') p^*(\psi'',s'',a'',r'') \cdots.
\end{align}
Thus any property of the original system is also a property of the generalized coordinate~system.


\end{appendices}

\bibliographystyle{apalike}
\bibliography{bibliography}
\end{document}